\documentclass[preprint,12pt,3p]{elsarticle}
\usepackage{amssymb}
\usepackage{amsmath, amsthm, amssymb, amsbsy}
\usepackage{amssymb}
\usepackage{amsfonts}
\DeclareMathAlphabet\mathbfcal{OMS}{cmsy}{b}{n}
\usepackage{amssymb}
\usepackage[skip=5pt]{caption}
\usepackage{tikz}
\usetikzlibrary{trees}
\usepackage{graphicx}
\usepackage{wrapfig}
\usepackage{booktabs,multirow}
\usepackage{hhline}
\usepackage[utf8]{inputenc}
\usepackage[english]{babel}
\usepackage{amsmath}
\usepackage{amsthm}
\usepackage{epstopdf}
\usepackage{bm}
\usepackage[english]{babel}
\usepackage{placeins}
\usepackage{makecell}
\usepackage{tabularx}
\usepackage{varioref}
\usepackage{caption}
\usepackage{subcaption}
\usepackage{mathrsfs}
\usepackage{float}
\usepackage{setspace}
\usepackage[utf8]{inputenc}
\usepackage{fancyhdr}
\usepackage{graphicx}
\usepackage{adjustbox,lipsum}
\usepackage{hyperref}
\usepackage{tabularx}
\usepackage{enumerate}
\usepackage{enumitem}
\usepackage[utf8]{inputenc}
\usepackage[T1]{fontenc}
\usepackage{caption}
\usepackage[utf8]{inputenc}
\usepackage[english]{babel}

\newtheorem{theorem}{Theorem}

\journal{arxiv.org}

\begin{document}

\begin{frontmatter}

\title{Classification of parameter spaces for reaction-diffusion systems on stationary domains}
%\tnotetext[label0]{This is only an example}

\author[label1,label2]{Wakil Sarfaraz\corref{cor1}\fnref{label3}}
\address[label1]{University of Sussex, School of Mathematical and Physical Sciences, Department of Mathematics, Pevensey 3, Brighton, BN1 9QH, UK}
%\address[label2]{Pevensey 3 5C15, Brighton, BN1 9RH, UK\fnref{label3}}

\cortext[cor1]{W.Sarfaraz and A.Madzvamuse are both corresponding authors}
%\fntext[label3]{I also want to inform about\ldots}
%\fntext[label4]{Small city}

\ead{wakilsarfaraz@gmail.com}
%\ead[url]{author-one-homepage.com}

\author[label1]{Anotida Madzvamuse}
%\address[label5]{University of Sussex}
\ead{A.Madzvamuse@sussex.ac.uk}

%\author[label1,label5]{A. Madzvamuse and W. Sarfaraz are both corresponding authors}
%\ead{author.three@mail.com}

\begin{abstract}
This paper explores the classification of parameter spaces for reaction-diffusion systems of two chemical species on stationary domains. The dynamics of the system are explored both in the absence and presence of diffusion. The parameter space is fully classified in terms of the types and stability of the uniform steady state. In the absence of diffusion the results on the classification of parameter space are supported by simulations of the corresponding vector-field and some trajectories around the uniform steady state. In the presence of diffusion, the main findings  are the quantitative analysis relating the domain-size with the reaction and diffusion rates and their corresponding influence on the dynamics of the reaction-diffusion system when perturbed in the neighbourhood of the uniform steady state.  Theoretical predictions are supported by numerical simulations both in the presence as well as in the absence of diffusion. Conditions on the domain size with respect to the diffusion and reaction rates are related to the types of diffusion-driven instabilities namely Turing, Hopf and Transcritical types of bifurcations. The first condition is an upper bound on the area of a rectangular domain in terms of the diffusion and reaction rates, which forbids the existence of Hopf and Transcritical types of bifurcations, yet allowing Turing type instability to occur. The second condition (necessary) is a lower bound on the domain size in terms of the reaction and diffusion rates to give rise to Hopf and Transcritical types of bifurcations as well as Turing instability. 
\end{abstract}
\begin{keyword}
%% keywords here, in the form: keyword \sep keyword
Reaction-diffusion systems \sep Dynamical systems\sep Bifurcation analysis\sep Stability analysis\sep Turing diffusion-driven instability \sep Hopf Bifurcation \sep Transcritical bifurcation \sep Parameter spaces
%% MSC codes here, in the form: \MSC code \sep code
%% or \MSC[2008] code \sep code (2000 is the default)
\end{keyword}
\end{frontmatter}

%%
%% Start line numbering here if you want
%%
% \linenumbers

%% main text
\section{Introduction}
\label{sec1}
Reaction-diffusion systems (RDSs) attract a significant degree of attention from researchers in applied mathematics \cite{paper3, paper4, paper5, paper6, paper8, paper9, paper10, paper11, paper12}, mathematical and computational biology \cite{paper2, paper7, paper23, paper24, paper27}, chemical engineering \cite{paper30, paper31, paper32} and so forth. Alan Turing was one of the first scientists to realise the significance of RDSs as a self-governing dynamical system \cite{paper1}, and showed that RDSs can be responsible for the emergence of spatial patterns in nature. A large number of scientists \cite{paper3, paper4, paper7, paper11, paper15, paper20, paper22}, since the publication of \cite{paper1}, have contributed to investigating RDSs, with various types of reaction kinetics. The most popular models of reaction kinetics explored in the literature are the {\it activator-depleted} model (also known as the Schnakenberg reaction kinetics) \cite{paper20, book1, paper17, paper4}, Meinhardt \cite{paper23, paper24} and Thomas \cite{paper25} reaction kinetic models. From a research perspective the study of RDSs is conducted through different types of approaches, one of which focuses on the local behaviour of the dynamics of the RDS near a uniform steady state, which in turn relates to the subject of stability analysis \cite{paper11, paper9, paper13, paper18, paper20, paper23} of RDSs using the stability matrix. Linear stability analysis offers a great deal of insight regarding the behaviour of RDSs in the neighbourhood of a uniform steady state of a particular system. The other usual approach is numerical computation of the actual solution of RDSs using finite element, finite difference and other numerical methods \cite{paper3, paper4, paper6, paper8, paper10, paper12, book5, book4}. Numerical solution of RDSs helps to visualise the evolution of the dynamics in time. Numerical computations of RDSs with non-linear reaction kinetics with a particular choice of parameters partially encapsulates the possible types of dynamics that a certain RDS can exhibit. In order to better understand the behaviour of RDSs, it is necessary to have prior knowledge on the classification of the parameter values. Madzvamuse \textit{et al}., in \cite{paper20, paper18} found regions of parameter space, corresponding to diffusion-driven instability with and without cross-diffusion respectively using the well-known {\it activator-depleted} reaction kinetics. Their approach to finding unstable regions in the parameter space is restricted to Turing instability only. One of the few complementary contributions from the present work is the application of a numerical method (exclusive to this paper) in order to obtain the full classification of parameter space.  The numerical method is also employed to solve the equations for the implicit curves forming the partitioning of the classification within the parameter space. Liu \textit{et al}., in \cite{paper28} attempted to find constraints on the parameters of RDSs with {\it activator-depleted} reaction kinetics that causes the system to exhibit Hopf and transcritical bifurcations. The proofs in \cite{paper28} are focused on the existence of bifurcation points for some theoretical constraints of parametrised variables of the actual parameters, with no relation between the domain size and reaction-diffusion rates. Comparing their work to the present study, our results are robust in the sense that we explicitly relate domain size to the reaction-diffusion rates. Using this relationship, the parameter space is classified for different types of bifurcations. Moreover, in the present work the parameter constraints are a consequence of the relationship between the domain size and the reaction-diffusion rates. An additional drawback in the analysis of \cite{paper28} is that their results are produced on parametrised variables of the model and not on the actual parameters of the equations, which makes their results applicable to non-realistic possibilities (negative values) of the actual parameters of the model. This drawback is effectively resolved in the current work as the analysis is conducted on the actual two-dimensional positive real parameter space, which in addition to confirming the existence of different bifurcation regions, it offers concrete quantitative classification of the parameter space that guarantees the dynamics of RDSs to exhibit these bifurcations.

The majority of RDSs in the literature \cite{paper4, paper5, paper12, paper15, paper22, thesis1} that exhibit spatial or temporal pattern, contain nonlinear terms in their reaction kinetics, which makes the mathematical and numerical analysis of such systems extremely challenging. With no closed analytical solutions, studies of the local behaviour of the systems are generally conducted by use of linear stability theory close to bifurcation points. Here, the behaviour of a system can be theoretically predicted in the neighbourhood of the uniform steady states. Linear stability analysis can help derive certain conditions on the reaction kinetics, which lead to causing instability in the dynamics of RDSs in the presence of diffusion. Numerous papers \cite{paper9, paper11, paper17, paper18, paper20} exist in literature that routinely apply linear stability analysis to RDSs.  The majority of the published work \cite{paper9, paper17} focuses on deriving linear stability conditions in terms of the characteristics of the stability matrix for diffusion-driven instability, lacking to explore what the numerical application of these conditions induce on the admissible choice of the parameter space for a certain RDS. Spatio-temporal pattern formation occurs when an RDS undergoes diffusion-driven instability \cite{paper1, book1, paper26}. This occurs when a uniform steady state which is stable in the absence of diffusion becomes unstable upon adding diffusion to the system. Diffusion-driven instability crucially depends on the values of diffusion and reaction rates of each reactant, however, more importantly it depends on the parameter choice of the reaction kinetics. In the current work the existing knowledge on the conditions for diffusion-driven instability in the literature is extended using a series of analytical and numerical techniques, to obtain new insights on the combined effects of diffusion and reaction rates, and in turn relating these to domain size of the evolution of the pattern. The detailed and quantitative analysis on the relationship between the domain size and diffusion-reaction rates in light of the parameter classification is an aspect that has not yet to-date received sufficient attention in the literature. The usual approach in selecting parameters for numerical computations of RDSs \cite{paper26, paper27} is based on the behaviour of RDSs in the absence of diffusion by use of trial and error or is based on previously published work, to observe instability caused by diffusion \cite{paper3, paper4, paper9, paper12, paper15}. The absence of a robust method to fully classify the parameter space for an RDS, creates an arguable platform for the importance of this work. Efficient analytical as well as computational methods are used to demonstrate the quantitative relationship between the domain size and diffusion rate for an {\it activator-depleted} RDS. The main findings of the present work, which relate the domain size to the diffusion and reaction rates, are presented in the form of theorems with rigorous mathematical proofs and these theoretical results are supported computationally by finite element numerical solutions corresponding to the {\it activator-depleted} RDS on fixed rectangular domains. For each numerical demonstration the relative error plots of the solutions for each successive time-step are presented to visualise the convergence of the numerical approximate solutions.

This article is therefore structured as follows. In Section \ref{absence} we carry out detailed theoretical linear stability analysis of the system (\ref{B}) in the absence of diffusion. Linear stability analysis is conducted by computing the stability matrix through which, the non-dimensional parameter space is derived and classified. Section \ref{partitionone} presents the methodology to compute the solutions of the partitioning curves for the classification on the parameter space. A combination of analytical and numerical methods using polynomials are applied. In Section \ref{presence}, the linear stability analysis of the system is conducted in the presence of diffusion and the parameter space is explored to understand the consequences of including diffusion in the analysis. The parameter space is further explored to find the change in regions of the parameter space through varying the non-dimensional diffusion coefficient. Section \ref{main} contains the main findings of this work which are presented in the form of two theorems. Each theorem is supported by finite element solutions of the model system. Section \ref{conclusion} presents conclusions, future directions and possible extensions of the current work. 

\section{Model equations}\label{absence}
The dynamics of two chemical species $u$ and $v$ in a coupled system of nonlinear parabolic equations is considered. Both species diffuse with independent rates and satisfy the well-known {\it activator-depleted} reaction-diffusion system in a closed two-dimensional rectangular domain denoted by $\Omega \subset \mathbb{R}^2$ with area $L_x\times L_y$, where $L_x$ and $L_y$ are the corresponding side lengths in the direction of $x$ and $y$ axes respectively. The reaction-diffusion system satisfying the {\it activator-depleted} model for $u$ and $v$ in its non-dimensional form \cite{paper20, book1, paper17, paper4} reads as
\begin{equation}
\begin{split}
\begin{cases}
   \frac{\partial u}{\partial t} =& \triangle u + \gamma (\alpha -u+u^2v), \\
   \frac{\partial v}{\partial t} =& d \triangle v + \gamma( \beta - u^2v),
\end{cases}
   \label{B}
      \end{split}
\end{equation}
where $d$, $\gamma$, $\alpha$ and $\beta$ are non-dimensional positive constants. In (\ref{B}) the non-dimensional parameter $d$ denotes the quantity $\frac{D_v}{D_u}$, where $D_u$ is the diffusion rate of the variable $u$ and $D_v$ is the diffusion rate of the variable $v$. The non-dimensional parameter $\gamma$ denotes the reaction rate, which is also known as the scaling parameter for the reaction kinetics. The boundary $\partial \Omega$ is subject to zero flux condition, which means on $\partial \Omega$ the chemical species $u$ and $v$ satisfy homogeneous Neumann boundary conditions in the form
\begin{equation}\label{zerofluxBCs}
\begin{split}
     \frac{\partial u}{\partial \bm{n}} = \frac{\partial v}{\partial \bm{n}}=0, \qquad \text{on} \quad (x,y) \in \partial \Omega, \quad t\geq 0,
     \end{split}
 \end{equation}
where $\bm{n}$ denotes the outward normal through $\partial \Omega$. For initial conditions the existence of some strictly positive quantity from each of these chemical concentrations in the domain is assumed, which is written as
 \begin{equation}\label{ICs}
 \begin{split}
     u(x,y,0) = u_0(x,y), \qquad v(x,y,0) = v_0(x,y),\qquad(x,y) \in \Omega, \quad t=0.
     \end{split}
 \end{equation} 
\subsection{Remark}
\noindent Despite the fact that all the results obtained  in this paper hold for general rectangular geometries of $\Omega$ i.e. with possibilities $L_x \neq L_y$, however for simplicity in the analysis $\Omega$ is considered as a square, which means $L_x = L_y$. The results can be readily extended to a rectangular case by taking the area of $\Omega$ as $L^2$, where $L = \max \{L_x,L_y\}$.

\subsection{Stability analysis in the absence of diffusion}\label{absenceone}
In the absence of diffusion the system (\ref{B}) takes the form of a set of ordinary differential equations of the form
\begin{equation}
\begin{split}
\begin{cases}
   \frac{d u}{dt} =& \gamma(\alpha -u+u^2v) = \gamma f(u,v),\\
\frac{d v}{dt} =& \gamma(\beta - u^2v) = \gamma g(u,v).
\end{cases}
\end{split}
\label{C}
\end{equation}
To analyse the stability of system (\ref{C}), it is necessary to compute its uniform steady state solution. Let $(u_s,v_s)$ denote the uniform steady state solution of the system (\ref{C}), then $(u_s,v_s)$ must simultaneously satisfy the system of nonlinear algebraic equations in the form
 \begin{equation}
 \begin{split}
    f(u_s,v_s)=&\alpha - u_s +u_s^2v_s =0, \qquad  g(u_s,v_s)=\beta - u_s^2v_s = 0.
    \end{split}
    \label{C1}
\end{equation}
The nonlinear algebraic system (\ref{C1}) admits a unique solution in the form 
\[
(u_s, v_s) = \Big(\alpha+\beta, \frac{\beta}{(\alpha+\beta)^2}\Big),
\]
which enforces a restriction on the parameters of the system such that $\alpha+\beta \neq 0$. Since both of these parameters resemble physical quantities, therefore, strictly positive, an appropriate interpretation of this restriction is that they both simultaneously cannot become zero.

\subsection{Stability matrix}
The stability of system (\ref{C}) is analysed by computing the Jacobian matrix \cite{book1} of (\ref{C1}) and conducting the stability analysis using the uniform steady state solution $(u_s, v_s)$, hence 
\[ \bm{J}|_{(u_s,v_s)}=\gamma \left[ \begin{array}{cc}
\frac{\partial f}{\partial u}  & \frac{\partial f}{\partial v} \\ 
\frac{\partial g}{\partial u}& \frac{\partial g}{\partial v} 
\end{array} \right]_{(u_s,v_s)}=\gamma \left[ \begin{array}{cc}
2uv-1  & u^2 \\ 
-2uv   & -u^2 
\end{array} \right]_{(u_s,v_s)}. 
 \]
Substituting the expressions in terms of $\alpha$ and $\beta$ for $(u_s,v_s)$ in $\bm{J}$ the matrix becomes 
\[
\bm{J}|_{(u_s,v_s)}=\gamma \left[ \begin{array}{cc}
\frac{\beta -\alpha}{\alpha +\beta}  & (\beta+\alpha)^2 \\ 
-\frac{2\beta}{\beta+\alpha}& -(\beta+\alpha)^2 
\end{array} \right],
\]
which is called the stability matrix \cite{book1, book7} for system (\ref{C1}).

\subsection{Parameter analysis}
Let 
\[
T(\alpha,\beta) =\gamma  \Big(\frac{\beta-\alpha-(\beta+\alpha)^3}{\beta+\alpha} \Big)\quad \text{and}\quad D(\alpha,\beta) =\gamma^2 (\alpha+\beta)^2,
\]
denote the trace and determinant of $\bm{J}$ respectively, then the characteristic polynomial for the eigenvalues $\lambda_{1,2}$ of $\bm{J}$ in terms of  $T(\alpha,\beta)$ and $D(\alpha,\beta)$ takes the form
\[
 \lambda^2 - T(\alpha,\beta)\lambda + D(\alpha,\beta) = 0.
\]
Hence, the two roots of this characteristic polynomial in terms of $T(\alpha,\beta)$ and $D(\alpha,\beta)$ are given by \begin{equation}
\lambda_{1,2} = \frac{1}{2} T(\alpha,\beta) \pm \frac{1}{2}\sqrt{T^2(\alpha,\beta)-4D(\alpha,\beta)}.
\label{character}
\end{equation}
Expression (\ref{character}) for the eigenvalues is studied through investigating the domain of $T$ and $D$, which is the positive real cartesian plane $(\alpha,\beta)\in\mathbb{R}_+^2$. The classification that would cause the uniform steady state $(u_s,v_s)$ to change stability and type  due to the selection of the choice of the parameter values $(\alpha,\beta)$ is explored by examining the sign of the real part \cite{book1, paper9, paper17, paper18} of $\lambda_{1,2}$.
For example the parameter space that makes the uniform steady state $(u_s,v_s)$ stable, is the simultaneous combined choice of $\alpha,\beta \in \mathbb{R}_+ $ that ensures the real part of the eigenvalues to be negative, which in turn is related to the discriminant of the roots expressed by (\ref{character}). The full parameter space is investigated, so that all the possible types of influences due to the choice of parameters $\alpha$ and $\beta$ on the stability and types of the uniform steady state $(u_s,v_s)$ are encapsulated. The classification is conducted based on $\lambda_{1,2}$ to be a complex conjugate in the first case, then in the second case the parameter space is analysed when $\lambda_{1,2}$ are real roots. In each case, the space is further classified into stable and unstable regions. In addition to this the partitioning curves, on which the steady state $(u_s,v_s)$ changes its type, are studied using a numerical technique that is employed exclusively and for the first time in the context of the present topic. 

\subsection{Analysis for the case of complex eigenvalues}
It is clear that the eigenvalues are a complex conjugate pair if and only if the discriminant is negative, which means the parameters $(\alpha,\beta)$ must satisfy the inequality
\begin{equation}\begin{split}
T^2(\alpha,\beta)-4D(\alpha,\beta)=\gamma^2 \Big(\frac{\beta-\alpha-(\beta+\alpha)^3}{\beta+\alpha}\Big)^2-4\gamma^2(\beta+\alpha)^2 <0.
\label{fiq}
\end{split}
\end{equation}
In order to find what region is satisfied by (\ref{fiq}), we must find the critical curves on which the expression on the left of (\ref{fiq}) is equal to zero, which means that the discriminant changes sign by moving across these curves in the plane $(\alpha,\beta) \in \mathbb{R}_+$. These curves can be found by solving the equation $T^2(\alpha,\beta)-4D(\alpha,\beta) = 0$, which is true for the choice of $(\alpha, \beta) \in \mathbb{R}^2_+$, satisfying 
\begin{equation}
\gamma^2 \Big(\frac{\beta-\alpha-(\beta+\alpha)^3}{\beta+\alpha}\Big)^2-4\gamma^2(\beta+\alpha)^2 =0.
\label{Ineq}
\end{equation}
Solving (\ref{Ineq}) implies finding the implicit curves that represent the critical region on the $(\alpha,\beta)\in \mathbb{R}_+$ plane, for which the discriminant in the expression for $\lambda_{1,2}$ is zero.  The solution to (\ref{Ineq}) provides the boundaries, for the region of the plane that results in $\lambda_{1,2}$ to be a complex conjugate pair. The choice of $(\alpha,\beta)$ on the curves satisfying (\ref{Ineq}) enforces the eigenvalues of the system to be repeated real values, therefore, on these curves the steady state $(u_s,v_s)$ becomes a star node, whose stability will be analysed in the Section \ref{absencereal} with real eigenvalues. The left hand-side of (\ref{Ineq}) can be factorised in the form $\phi(\alpha,\beta)\psi(\alpha,\beta)$, which provides the equations of the two implicit curves $\phi(\alpha,\beta)=0$ and $\psi(\alpha,\beta)=0$ that determine the boundaries of the region corresponding to complex eigenvalues $\lambda_{1,2}$, where $\phi$ and $\psi$ are respectively given by
\begin{equation}
\begin{split}
\begin{cases}
\phi(\alpha,\beta) = \beta-\alpha -(\beta+\alpha)^3-2(\beta+\alpha)^2,\\
\psi(\alpha,\beta) = \beta-\alpha -(\beta+\alpha)^3+2(\beta+\alpha)^2.
\end{cases}
\label{imp}
\end{split}
\end{equation}

\subsubsection{Solving the partitioning curves}\label{partitionone}
A mesh is constructed on a square domain $D=[0,\alpha_{max}]\times[0,\beta_{max}]$, which is discretised by $N$ points in both directions of $\alpha$ and $\beta$, where $N$ is a positive integer. This constructs a square mesh of $(N-1)\times(N-1)$ cells, each of size $\frac{\alpha_{max}}{N} \times \frac{\beta_{max}}{N}$, with $N^2$ points in $D$. 
To find the implicit solution for (\ref{imp}), at every mesh point in the direction $\alpha $, the roots of the cubic polynomial in $\beta$ namely $\phi(\alpha_i,\beta)=0$, denoted by $\phi_i(\beta)=0$ are computed using the Matlab command `$roots$'. It is worth noting that for every fixed $\alpha_i$, one obtains $\phi_i(\beta)$, to be a cubic polynomial of degree 3 in $\beta$ of the form
\[
\phi_i(\beta) = C^i_0 + C^i_1 \beta + C^i_2 \beta^2 + C^i_3 \beta^3,
\]
where the coefficients are given by 
$C^i_0 = -\alpha_i-\alpha_i^3 -2\alpha_i^2$, $C^i_1 = 1-4\alpha_i -3\alpha_i^2$, $ C^i_2 = -3\alpha_i-2$, and $ C^i_3 = -1.$
Similarly for each $i$ in the direction of $\alpha$, there are $N$ cubic polynomials in $\beta$ satisfying $\psi_i(\beta)=0$, with $\psi_i$ of the form
\[
\psi_i(\beta) = C^i_0 + C^i_1 \beta + C^i_2 \beta^2 + C^i_3 \beta^3,
\]
where $C^i_0 = -\alpha_i-\alpha_i^3 +2\alpha_i^2$, $C^i_1 = 1+4\alpha_i -3\alpha_i^2$, $C^i_2 = -3\alpha_i+2$, and $C^i_3 = -1$.
Each of the equations $\phi_i(\beta)=0$ and $\psi_i(\beta)=0$ have at most three roots namely $(\beta_1,\beta_2,\beta_3)$ for every fixed $\alpha_i$, which means for every fixed $\alpha_i$ the three points namely $(\alpha_i, \beta_j)$ for $j=1,2,3$ are the three points that lie on the implicit curve given by (\ref{imp}). Since $(\alpha,\beta)$ are positive real parameters in the system, therefore, at every fixed $\alpha_i$, only the positive real roots $\beta_j$ are considered and any root that is either real negative or complex is ignored. 
%For all the fixed points in the approximate interval $\alpha_i \in [0 , 0.06]$, $\phi_i(\beta)$ has three distinct positive real roots, which is why if a vertical line test is applied to Figure \ref{fig2} (a) in the interval $\alpha \in [0 , 0.06]$ the vertical line will cross the solution curve for (\ref{imp}) at three points. The solution curves that satisfy (\ref{imp}) are the blue curves namely $c1$ and $c3$ in Figure \ref{fig2}.  When $\alpha_i \in (0.06, 1.315]$, is chosen then the cubic polynomial $\phi_i(\beta)=0$ has two positive real roots and the third root is either repeated or negative real, which is why the vertical line test crosses the $c1$ curve twice in this interval. Choosing $\alpha_i>1.315$ causes the polynomial $\phi_i(\beta)$ to have no positive real root, which means all the roots are either imaginary or they are real negative. Since the blue curves $c1$ and $c3$ in Figure \ref{fig2} are plots of positive real roots of $\phi_i(\beta)=0$ only, and the imaginary  or negative roots are ignored, therefore it is bounded in the direction of $\alpha$ by the approximate value of $1.315$. Similarly for $\beta >2.41$ there are no positive real roots that satisfies $\phi(\alpha, \beta_i) = 0$, therefore the blue curve is bounded in the direction of $\beta$ by the approximate value 2.41. 

Solving (\ref{Ineq}) alone serves to determine the boundaries of the complex region, but does not tell us, which side of these implicit curves correspond to real eigenvalues and which side corresponds to complex eigenvalues. This can be decided, by finding another curve on which the eigenvalues are purely imaginary with zero real parts. Since the real part of the eigenvalues is given by $\frac{1}{2}T(\alpha,\beta)$ and the determinant $D(\alpha,\beta)=(\alpha+\beta)^2$ is strictly positive, therefore, the choice of $(\alpha,\beta)$ that solves the equation
\begin{equation}
T(\alpha,\beta) = 0  \qquad \iff \qquad \frac{\beta-\alpha-(\beta+\alpha)^3}{\beta+\alpha}=0,
\label{T}
\end{equation}
will ensure that the eigenvalues of the system are purely imaginary. Solving (\ref{T})  is equivalent to finding the set of $(\alpha,\beta) \in \mathbb{R}^2_+$ such that the equation $h(\alpha,\beta)=0$ is true, where $h$ is given by
\[
h_i(\beta) = C_0^i+C_2^i\beta+C_2^i\beta^2+C_3^i\beta^3,
\]
with coefficients
$C_0^i =-\alpha_i-\alpha_i^3,$
$C_1^i =1-\alpha_i^2,$
$C_2^i =-3\alpha_i,$ and 
$C_3^i =-1.$

The curve satisfying equation (\ref{T}) forms the boundary for the positive and negative real parts of the eigenvalues $\lambda_{1,2}$ when they are complex conjugate pair, because on $h_i(\beta)=0$ the real part of the eigenvalues are zero and hence changes the sign by moving across it. The positive real combination of $(\alpha,\beta)$ that satisfies (\ref{T}) is given by the red $c2$ curve in Figure \ref{fig2}. The location of the $c2$ curve also indicates that the region between the curves $c1$ and $c3$ results in that $\lambda_{1,2}$ are complex numbers. If $(\alpha,\beta)$ are chosen from values on curves $c1$ and $c3$, then the eigenvalues will be repeated real roots, since on these curves the discriminant is zero. However if $(\alpha,\beta)$ are chosen from the values on the curve $c2$ then the eigenvalues are purely imaginary of the form
%\[
$\lambda_{1,2} = \pm i\sqrt{4D(\alpha,\beta)} =\pm i2(\alpha+\beta)$.
%\]
This implies that the steady state $(u_s,v_s)$ is a centre with the system exhibiting periodic oscillations around the uniform steady state. It is worth noting that $h_i(\beta) =0$ is a sufficient condition for $(u_s,v_s)$ to be a centre, due to the strictly positive expression for $D(\alpha,\beta) =(\alpha+\beta)^2$ in the discriminant. Figure \ref{fig2}, was simulated for $\alpha_{max}=\beta_{max}= 5$ and $N = 5000$, which means the curve $c1$, $c2$ and $c3$ are respectively formed from the positive real roots of $5000$ cubic polynomials of $\psi_i(\beta)=0$, $h_i(\beta)=0$ and $\phi_i(\beta)=0$ for $i=1,...,5000$.

The curves $c1$ and $c3$ in Figure \ref{fig2} (a) are the critical boundaries, at which the type of the steady state $(u_s,v_s)$ changes from node to a spiral or vice versa, depending on whether one enters or exists the region between $c1$ and $c3$. The real part of $\lambda_{1,2}$ is studied to find the stability classification of the parameter space, when $\lambda_{1,2}$ are complex conjugate pair. Knowing that the region bounded by $c1$ and $c3$ curves in Figure \ref{fig2} (a) results in $\lambda_{1,2}$ to be complex, the next step is to find which part of this region corresponds to positive real part of the complex numbers $\lambda_{1,2}$, and which part of this region corresponds to negative real part of $\lambda_{1,2}$. In order to determine their stability we must analyse the real part of $\lambda_{1,2}$, simultaneously with the assumption that the roots are a pair of complex conjugate numbers, which is given by
\begin{equation}
   T(\alpha,\beta)=\frac{\beta-\alpha-(\beta+\alpha)^3}{\beta+\alpha}.
   \label{Ineq2}
\end{equation} 
If the sign of (\ref{Ineq2}) is negative in the region between the curves $c1$ and $c3$ in Figure \ref{fig2} (a), then $(u_s,v_s)$ is a stable spiral, if it is positive then it is an unstable spiral. This means that the steady state $(u_s,v_s)$ is a stable spiral if $(\alpha,\beta)$ simultaneously satisfy (\ref{fiq}) and the inequality
$ \beta -\alpha < (\alpha+\beta)^3$ holds,
which corresponds to the green region A shown in Figure \ref{fig2} (b). Similarly, the steady state $(u_s,v_s)$ is an unstable spiral if parameters $(\alpha,\beta)$ simultaneously satisfy (\ref{fiq}) and the inequality $ \beta -\alpha > (\alpha+\beta)^3$,
which is the yellow region B shown in Figure \ref{fig2} (b). 
It is worth noting that the trace $T(\alpha,\beta)$ for positive values of $\alpha$ and $\beta$ is bounded by the value of 1 [see section 4. Theorem \ref{theorem2}], which means
%\[
$\frac{\beta-\alpha-(\alpha+\beta)^3}{\beta+\alpha}\leq 1$ for all $\alpha,\beta\in \mathbb{R}_+$.
%\]
\subsubsection{Numerical demonstration of complex eigenvalues}
The numerical package for ordinary differential equations {\it ode45} in MATLAB was employed to simulate and visualise the local behaviour of the system (\ref{C}) near the uniform steady state $(u_s,v_s)$. Each simulation was run up to a final time $T=5$, starting at the initial time $T=0$, with time step $\tau = 5\times 10^{-2}$. Parameter values were chosen from the regions $A$, $B$ and on the curve $c2$ from Figure \ref{fig2} (b) and Figure \ref{fig2} (a) respectively, which are then plotted with 8 trajectories around the uniform steady state $(u_s,v_s)$. Figure \ref{phase} shows that the behaviour of the trajectories around the uniform steady state in each case is in agreement with the theoretically predicted type of $(u_s,v_s)$. The summary of the chosen numerical values for these simulations are all summarised in Table \ref{Table1} and these are also indicated in the relevant captions for each figure. The eigenvalues of the stability matrix for each choice of the parameters and at the corresponding uniform steady state are presented in the corresponding captions. In each simulation the vertical axes corresponds to the variable $v$ and the horizontal axes corresponds to the variable $u$. The red trajectories show the interactive behaviour of $u$ and $v$ in relation to the corresponding steady state.
\begin{figure}[H]
 \centering
 \small
  \begin{subfigure}[h]{.495\textwidth}
    \centering
 \includegraphics[width=\textwidth]{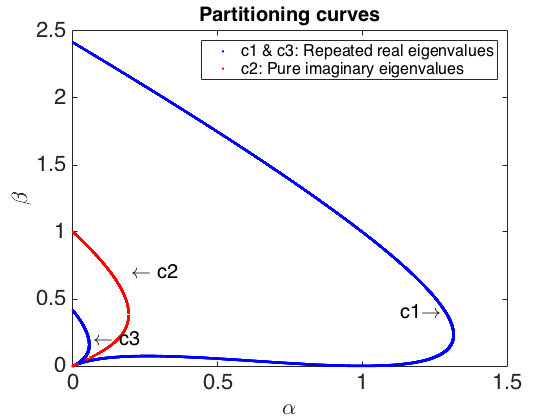}
 \caption{The choice for $\alpha$ and $\beta$ on curves\\ $c1$ and $c3$ make $(u_s,v_s)$ to be star node,\\ whereas $\alpha$ and $\beta$ on the curve $c2$ make  \\$(u_s,v_s)$ to be a centre with periodic \\oscillations.}
 \label{paramfig8}
 \end{subfigure}
 \begin{subfigure}[h]{.495\textwidth}
 \centering
    \includegraphics[width=\textwidth]{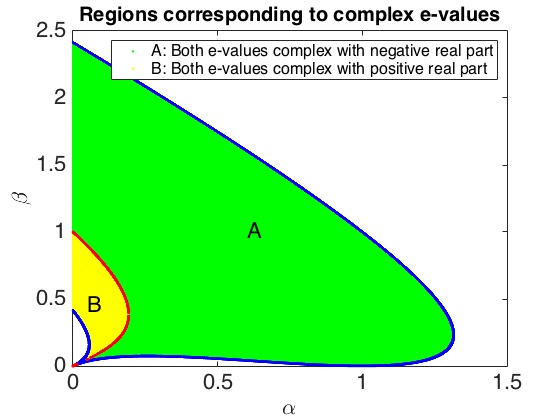}
  \caption{The choice for $\alpha$ and $\beta$ in the green region $A$ results in $(u_s,v_s)$ being a stable spiral, whereas $\alpha$ and $\beta$ from the yellow region $B$ results in $(u_s,v_s)$ being an unstable spiral.}
  \label{paramfig2}
 \end{subfigure}
 \caption{Stable and unstable regions when the steady state $(u_s,v_s)$ is a spiral stationary point.}
  \label{fig2}
\end{figure}
\begin{figure}[!ht]
 \centering
 \small
  \begin{subfigure}[h]{0.328\textwidth}
 \centering
    \includegraphics[width=\textwidth]{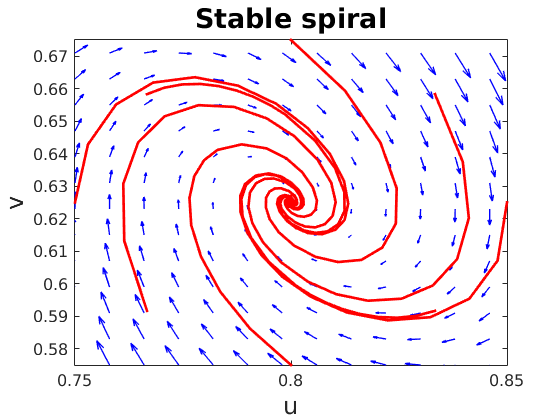}
  \caption{Parameters $(\alpha,\beta)=\\(0.4,0.4)$, resulting in the\\ uniform steady state \\$(u_s,v_s)=(0.8,0.625)$  \\and eigenvalues $\lambda_{1,2} =\\ -0.32\pm0.73i$.}
  \label{stablespiral}
 \end{subfigure}
  \begin{subfigure}[h]{0.328\textwidth}
 \centering
 \includegraphics[width=\textwidth]{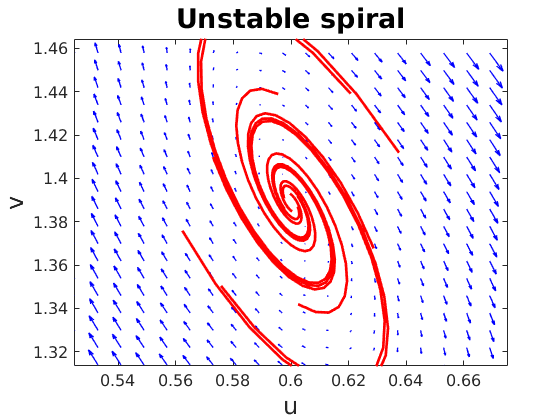}
  \caption{Parameters $(\alpha,\beta)=\\(0.1,0.5)$, resulting in the\\ uniform steady state \\$(u_s,v_s)=(0.6, 1.39)$  \\and eigenvalues $\lambda_{1,2}=\\ 0.15\pm0.58i$.}
  \label{unstablespiral}
 \end{subfigure}
  \begin{subfigure}[h]{0.328\textwidth}
 \centering
    \includegraphics[width=\textwidth]{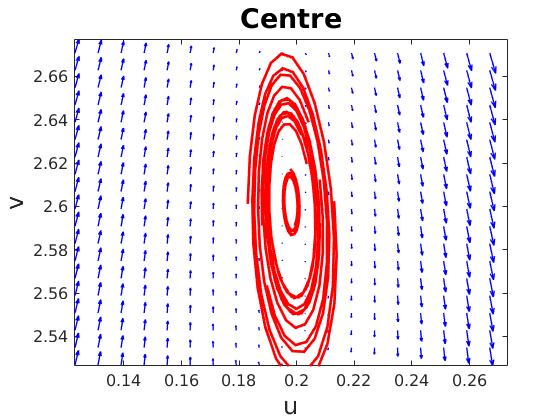}
  \caption{Parameters $(\alpha,\beta)=\\(0.096,0.102)$, resulting in the\\ uniform steady state \\$(u_s,v_s)=(0.198,2.6)$  \\and eigenvalues $\lambda_{1,2}=\\ \pm0.20 i$.}
  \label{centre}
 \end{subfigure}
 \caption{Phase-plane diagrams of the model system in the absence of diffusion for types of steady states with complex $\lambda_{1,2}$, computed correct up to 2 decimal places and $i = \sqrt{-1}$.}
  \label{phase}
\end{figure}

\subsubsection{Interpretation of the dynamics for the case of complex eigenvalues}
The physical interpretation of this classification is that, whenever the choice of $\alpha$ and $\beta$ is taken from region $B$ in Figure \ref{fig2} then the dynamics of the perturbed system near the constant steady state $(u_s,v_s)$ will spirally move away as time grows, which resembles the shape of an outward spreading spiral. This behaviour is shown in Figure \ref{phase} (b). However if the parameter values are chosen from region $A$ then the dynamics of the perturbed system will spirally move in towards the constant stable steady state $(u_s,v_s)$, therefore, any trajectory near the steady state forms the shape of spiral that is uniformly moving towards its centre, which is the constant steady state namely the point $(u_s,v_s)$. This behaviour is shown in Figure \ref{phase} (a). If the values of the parameters are chosen from those on the curve $c_2$, then the dynamics of the system behaves in such a way that it neither goes away from the steady state, and nor does it reach the steady state, instead the trajectories of the system move on fixed orbits around the uniform steady state $(u_s,v_s)$, forming either circles or ellipse. This behaviour of $u$ and $v$ is shown in Figure \ref{phase} (c). 
\subsection{Analysis for the case of real eigenvalues}\label{absencereal}
When both eigenvalues of the system are real numbers at the steady state $(u_s,v_s)$ then the uniform steady state becomes a node. If $\lambda_{1,2}$ is a positive repeated root, then the steady state is an unstable star node. If $\lambda_{1,2}$ is a negative repeated root, then the steady state is a stable star node. Parameter space resulting in $(u_s,v_s)$ to be a node can be analysed through studying the sign of $\lambda_{1,2}$, when both eigenvalues are real. This consequently means that the discriminant has to satisfy the inequality 
\begin{equation}
   T^2(\alpha,\beta)-4D(\alpha,\beta) = \Big(\frac{\beta-\alpha-(\beta+\alpha)^3}{\beta+\alpha}\Big)^2-4(\beta+\alpha)^2 \geq 0. 
   \label{real}
\end{equation}
Similarly, to the previous section, treating the equal case of (\ref{real}) first  will provide the boundaries between the real and complex regions for $\lambda_{1,2}$. These are $c1$ and $c3$ curves in Figure \ref{fig2} from the previous section. On $c1$ and $c3$ it was concluded that $\lambda_{1,2}$ are repeated real roots, which are precisely 
%\[
$\lambda_{1,2} = \frac{1}{2} \left( \frac{\beta-\alpha-(\beta+\alpha)^3}{\beta+\alpha}\right)$.
%\]
In order to determine the stability of $(u_s,v_s)$, when $(\alpha,\beta)$ are on $c1$ and $c3$ blue curves in Figure \ref{fig2}, it is required to find the classification of this curve on which the sign of the repeated eigenvalue $\lambda_{1,2}$ is positive or negative.
The steady state $(u_s,v_s)$ is a star, if $\lambda_{1,2}$ is repeated real values, which means that the discriminant must be zero
\begin{equation}
   T^2(\alpha,\beta)-4D(\alpha,\beta) = \Big(\frac{\beta-\alpha-(\beta+\alpha)^3}{\beta+\alpha}\Big)^2-4(\beta+\alpha)^2 = 0. 
   \label{equal}
\end{equation}
The steady state $(u_s,v_s)$ is a stable star if the choice of parameters $(\alpha,\beta)$ simultaneously satisfy (\ref{equal}) and the inequality
$
\beta -\alpha < (\alpha+\beta)^3,
$
which is the $c1$ curve in Figure \ref{fig4} (a) in blue colour.
The steady state $(u_s,v_s)$ is unstable star if the choice of parameters $(\alpha,\beta)$ simultaneously satisfy (\ref{equal}) and the inequality
$
\beta -\alpha > (\alpha+\beta)^3,
$
which is the $c3$ curve in Figure \ref{fig4} (a) in red colour.

The remaining region outside $c1$ and $c3$ curves in Figure \ref{fig4} (a) can be classified into two parts, one where $(\alpha,\beta)$ satisfy (\ref{real}) and yet both $\lambda_{1,2}$ are distinct negative real values. This means that $(\alpha,\beta)$ must satisfy the inequality 
\[
\lambda_{1,2}=\frac{1}{2}\frac{\beta-\alpha-(\beta+\alpha)^3}{\beta+\alpha} \pm \sqrt{\Big(\frac{\beta-\alpha-(\beta+\alpha)^3}{\beta+\alpha}\Big)^2-4(\beta+\alpha)^2} < 0,
\]
such that both eigenvalues are negative real. This is true if $(\alpha,\beta)$ satisfy the inequality
\begin{equation}
\frac{1}{2}\frac{\beta-\alpha-(\beta+\alpha)^3}{\beta+\alpha} + \sqrt{\Big(\frac{\beta-\alpha-(\beta+\alpha)^3}{\beta+\alpha}\Big)^2-4(\beta+\alpha)^2} < 0.
\label{bothneg}
\end{equation}
The region in the $(\alpha,\beta)\in \mathbb{R}_+$ plane satisfying (\ref{bothneg}) is denoted by A in Figure \ref{fig4} (b) shaded in green colour. The remaining region to analyse is the region outside the $c3$ curve in Figure \ref{fig4} (a). Any combination of  $(\alpha,\beta)$ from this region corresponds to the remaining 2 cases, in which either one of $\lambda_1$ or $\lambda_2$ is positive, causing $(u_s,v_s)$ to be a saddle point which is unstable by definition. The other case is if $\lambda_{1,2}$ are both positive real values, which makes $(u_s,v_s)$ to be an unstable node. The region satisfying the criteria for a saddle point or unstable node is indicated by B in Figure \ref{fig4} (b) in yellow colour. It is worth noting that the steady state $(u_s,v_s)$ with either one or both positive real eigenvalues is by definition a saddle point or an unstable node, respectively. Under the current classification these both types fall under one category namely unstable node. The fact that one of the eigenvalues is real positive, hence as time grows very large, the behaviour of the solution is similar to that of an unstable node.  
\begin{figure}[H]
 \centering
 \small
  \begin{subfigure}[h]{.495\textwidth}
    \centering
 \includegraphics[width=\textwidth]{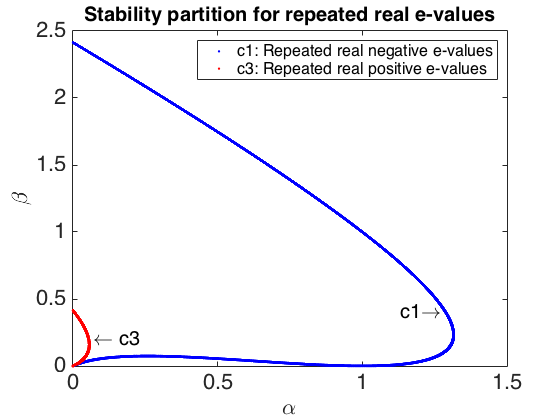}
 \caption{The combination of $(\alpha,\beta)$ from $c1$, $c2$ \\(blue  and red curves) results in $(u_s,v_s)$\\ being an unstable star and a stable star \\  respectively, which are both of type node.}
 \label{paramfig7}
 \end{subfigure}
 \begin{subfigure}[h]{.495\textwidth}
 \centering
    \includegraphics[width=\textwidth]{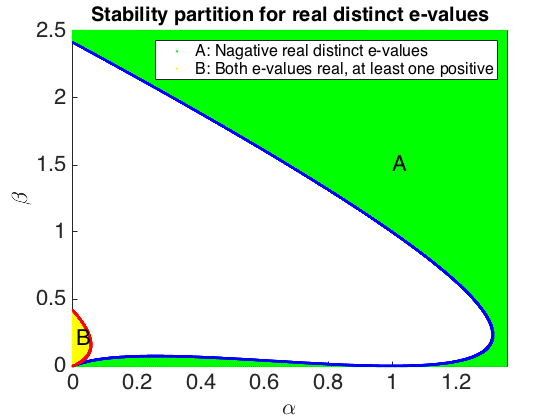}
  \caption{The combination of $(\alpha,\beta)$ from the yellow $B$ and green $A$ regions results in $(u_s,v_s)$ being an unstable and a stable node respectively.}
  \label{paramfig9}
 \end{subfigure}
 \caption{Stable and unstable regions when the steady state $(u_s,v_s)$ is a node.}
  \label{fig4}
\end{figure}
\subsubsection{Numerical demonstration for the case of real eigenvalues}
The system (\ref{C}) is numerically simulated using {\it ode45} package in MATLAB to visualise the behaviour of trajectories in the neighbourhood of $(u_s,v_s)$, when $\lambda_{1,2}$ are real values. The final time and time steps are chosen exactly the same as for complex eigenvalues case in the previous section. Each simulation is tested on phase plane diagram, to observe the trajectories of solution. In each case the steady state $(u_s,v_s)$ is computed and 8 trajectories around the steady state whose behaviour is governed by the relevant vector-field are plotted.
\begin{figure}[!ht]
 \centering
 \small
  \begin{subfigure}[h]{0.49\textwidth}
    \centering
 \includegraphics[width=\textwidth]{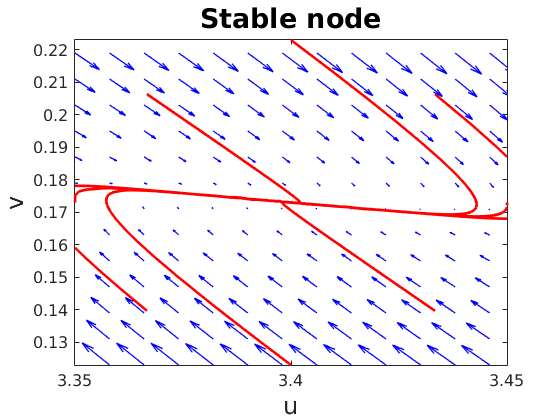}
 \caption{Parameters $(\alpha,\beta)=(1.4,2.0)$,\\ resulting in the uniform steady state \\$(u_s,v_s)=(3.4,0.173)$ and eigenvalues \\ $\lambda_{1,2}= (-10.26, -1.13)$.}
 \label{stablenode}
 \end{subfigure}
 \begin{subfigure}[h]{0.49\textwidth}
 \centering
    \includegraphics[width=\textwidth]{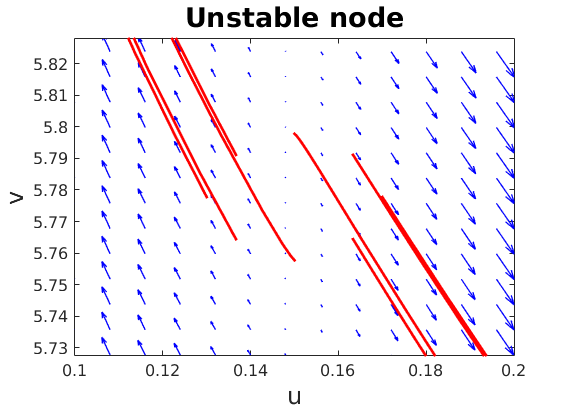}
  \caption{Parameters $(\alpha,\beta)=(0.02,0.13)$,\\ resulting in the uniform steady state\\ $(u_s,v_s)=(0.15,5.78)$  and eigenvalues\\ $\lambda_{1,2}= (0.03, 0.68)$.}
  \label{unstablenode}
 \end{subfigure}
  \begin{subfigure}[h]{0.49\textwidth}
 \centering
    \includegraphics[width=\textwidth]{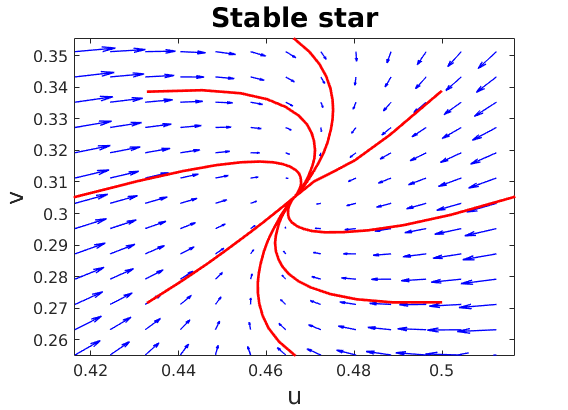}
  \caption{Parameters $(\alpha,\beta)=(0.40,0.066)$,\\ resulting in the uniform steady state \\$(u_s,v_s)=(0.8,0.625)$  and eigenvalues\\ $\lambda_{1,2} = (-0.47,-0.47)$.}
  \label{stablestar}
 \end{subfigure}
  \begin{subfigure}[h]{0.49\textwidth}
 \centering
    \includegraphics[width=\textwidth]{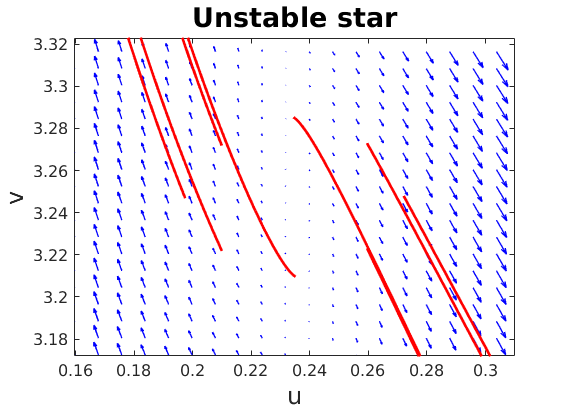}
  \caption{Parameters $(\alpha,\beta)=(0.06,0.18)$,\\ resulting in the uniform steady state \\$(u_s,v_s)=(0.23,3.25)$  and eigenvalues\\ $\lambda_{1,2} = (0.20,0.20)$.}
  \label{unstablestar}
 \end{subfigure}
 \caption{Phase-plane diagrams of the model system in the absence of diffusion characterising the stability of the steady states for real $\lambda_{1,2}$.}
  \label{fig6}
\end{figure}
\subsubsection{Interpretation of the dynamics for the case of real eigenvalues}
If the values of $\alpha$ and $\beta$ are selected from those on the curves $c_1$ or $c_3$ then the dynamics of the perturbed system will either radially move away or towards the point of the constant steady state $(u_s,v_s)$, forming the shape of a star. If parameters are selected from the curve $c_1$ in Figure \ref{fig4} (b), then the eigenvalues are positive real repeated roots. The dynamics of the perturbed system near the constant steady steady state will radially move outward (see Figure \ref{fig6} (d)). If the values of $\alpha$ and $\beta$ are chosen from those lying on the $c_3$ curve in Figure \ref{fig4} (b), then the local dynamics of the system near the constant steady state radially move inward towards the point $(u_s,v_s)$ (see Figure \ref{fig6} (c)). Similarly, if $\alpha$ and $\beta$ are chosen from region $A$ in Figure \ref{fig4} (b), it will result in $\lambda_{1,2}$ to be a pair of distinct negative real roots. This means for the dynamics of the system to move towards the stable constant steady state. The trajectories near by the steady state in this case move towards a single point namely $(u_s,v_s)$. If parameters are however chosen from region $B$ in Figure \ref{fig4} (b), this results in the constant steady state to be an unstable node, it means the behaviour of trajectories near the steady state are expected to behaviour similarly to those corresponding to the curve $c_2$ (see Figure \ref{fig6} (a)).

\subsection{Summary of parameter classification in the absence of diffusion}\label{fullclass}
In general, the parameter space that influences the nature of the uniform steady state $(u_s,v_s)$ can be categorised into four different regions, which are separated by three curves. This classification is in principle equivalent to that presented in \cite{book1}, however the distinction with the current work is that for the specific system given by (\ref{C}) the full parameter plane is classified subject to the proposed theory given in \cite{book1}. Each region is characterised and the behavior of the uniform steady state  is established by choosing parameters $\alpha$ and $\beta$ from the region. Figure \ref{fig5} (a) summarises the full classification of the parameter space in the absence of diffusion for the uniform steady state $(u_s,v_s)$. It can be noted that in terms of stability only and irrespective of the type of the steady state, the full parameter space namely the $(\alpha,\beta) \in \mathbb{R}_+$ plane can be classified into two regions namely stable or unstable, this is the partition established by the location of the yellow curve in Figure \ref{fig5} (a). Stability classification irrespective of the type of steady state is explicitly presented in Figure \ref{fig5} (b). The purpose of stability partition will prove beneficial in the next section when the diffusion-driven instability is explored. Because diffusion-driven instability as the name suggests, is to explore how a stable steady state becomes unstable for the same choice of parameters $(\alpha,\beta)$, when diffusion is added to the system. Therefore, with the help of stability partition, one is able to observe the change in the location of the yellow curve in Figure \ref{fig5} (b).  Table \ref{Table1} shows the values of the parameters $(\alpha,\beta)$ and the values of the unique steady state $(u_s,v_s)$ with the relative eigenvalues that corresponds to the summary of full classification of parameter space presented in Figure \ref{fig5} (a), and these parameter values are used in the numerical simulations. 
%\begin{figure}[!ht]
\begin{figure}[H]
 \centering
 \small
  \begin{subfigure}[h]{.49\textwidth}
    \centering
 \includegraphics[width=\textwidth]{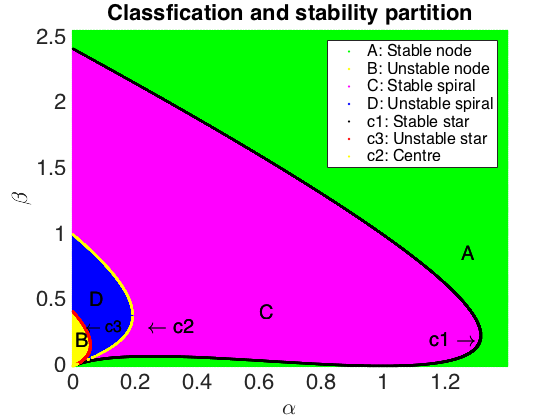}
 \caption{Parameter space classification that \\results in the steady state $(u_s,v_s)$ to be \\of different types, and their corresponding\\ stability as indicated in the legend.}
 \label{paramclass}
 \end{subfigure}
 \begin{subfigure}[h]{.49\textwidth}
 \centering
    \includegraphics[width=\textwidth]{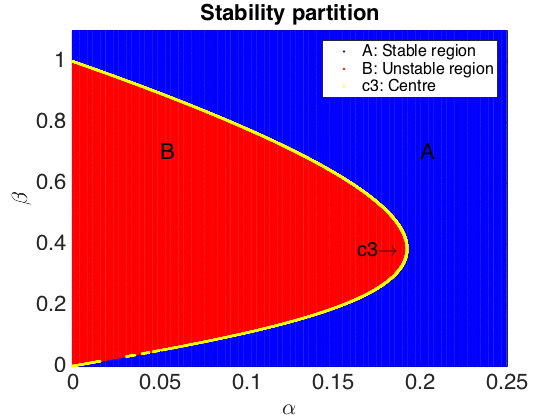}
  \caption{Parameters from red region make $(u_s,v_s)$ unstable, whereas any combination of $(\alpha,\beta)$ from the blue region make $(u_s,v_s)$ stable,  and parameters lying on  the yellow curve results in $(u_s,v_s)$ being a centre.}
  \label{paramstable}
 \end{subfigure}
 \caption{The full parameter classification for the stability and types of the steady state $(u_s,v_s)$.}
  \label{fig5}
\end{figure}

\begin{table}[ht]
%\centering
%\small
%\tabcolsep=0.2cm
\noindent\adjustbox{max width=\textwidth}{
\begin{tabular}{|c |c |c |c|c|c|}
\hline
 Name (SS)& $(\alpha,\beta)$& $(u_s,v_s)$&$\lambda_{1,2}$& Fig \ref{fig5}(a)& Phaseplane \\
\hline
\cline{1-2}
Stable node & $(1.400, 2.000)$ & $(3.400, 0.173)$  & (-10.26, -1.13) & Region A & Fig \ref{fig6} (a)\\
Unstable node & $(0.020, 0.130)$ & $(0.150, 5.780)$ & (0.03, 0.68)& Region B & Fig \ref{fig6} (b)\\
Stable spiral	& $(0.400, 0.400)$ & $(0.800, 0.625)$   & $ -0.32\pm0.73i$ &Region C & Fig \ref{phase} (a)\\
Unstable spiral  & $(0.100, 0.500)$ & $(0.600, 1.390)$ & $0.15\pm0.58i$  &Region D &Fig \ref{phase} (b)\\
Stable star  & $(0.400, 0.066)$ & $(0.470, 0.310)$   &  (-0.47, -0.47) &Curve $c1$ &Fig \ref{fig6} (c)\\
Unstable star  & $(0.060, 0.180)$ & $(0.230, 3.250)$ &  (0.20, 0.20) & Curve $c3$ &Fig \ref{fig6} (d)\\
Centre & $(0.096, 0.102)$& $(0.198, 2.600)$& $\pm0.20 i$ &Curve $c2$&Fig \ref{phase} (c)\\
\hline
\end{tabular}}
\caption{Table showing the summary of all types of steady states and their corresponding phase plane diagrams with reference to the parameter space on $(\alpha,\beta)\in \mathbb{R}_+$ plane.}
\label{Table1}
\end{table}

\section{Stability analysis in the presence of diffusion}\label{presence}
It is intuitive to understand that normally diffusion serves to enhance spatial homogeneity of concentration gradients \cite{book1, paper26}, however, in the case of reaction-diffusion systems, it has the opposite effect in that the uniform steady state which was stable in the absence of diffusion becomes unstable in its presence \cite{thesis1, thesis2, book1, paper26}. In this work, the mathematical rigor of how reaction and diffusion together can give rise to instability is investigated, which is responsible for diffusion-driven pattern formation. One of the explicit aims of the current work is to investigate the consequences of the conditions of diffusion-driven instability on the classification of the parameter space in the presence of diffusion. Furthermore, we want to explore quantitative relationship between the domain size and the types of diffusion-driven instability. 

\subsection{Linear stability analysis}
 The uniform steady state $(u_s,v_s)$ as previously mentioned satisfies the system with diffusion \eqref{B} and the zero-flux boundary conditions \eqref{zerofluxBCs}, therefore, without loss of generality, $(u_s,v_s)$ is considered a steady state of the system in presence of diffusion as well.
System (\ref{B}) is first linearised by introducing new variables namely $\bar{u}$ and $\bar{v}$ in such a way that they are perturbed slightly from the steady state $(u_s,v_s)$, so that their relationships to $u$ and $v$ are 
%\[\begin{split}
$(\bar{u}, \bar{v})= (u - u_s, v - v_s)$.
%\end{split}\]
For $u$ and $v$, the perturbed variables $\bar{u}+u_s$ and $\bar{v}+v_s$ are substituted in (\ref{B}) and expanded  using Taylor expansion for functions of two variables to obtain a linearised system which can be written in matrix form as
\begin{equation}
\frac{\partial}{\partial t}\left[\begin{array}{c}
     \bar{u}  \\
     \bar{v} 
\end{array}\right]=\left[\begin{array}{cc}
     1&0  \\
     0&d 
\end{array}\right]\left[\begin{array}{c}
     \triangle \bar{u} \\
     \triangle \bar{v} 
\end{array}\right]+\left[\begin{array}{cc}
     \frac{\partial f}{\partial u}(u_s,v_s)& \frac{\partial f}{\partial v}(u_s,v_s)  \\
     \frac{\partial g}{\partial u}(u_s,v_s)&\frac{\partial g}{\partial u}(u_s,v_s)
\end{array}\right]\left[\begin{array}{c}
     \bar{u} \\
     \bar{v} 
\end{array}\right].
\label{vect}
\end{equation}
In order to complete the linearisation in the presence of diffusion, it is necessary to find the eigenfunctions of the Laplace operator, that satisfies the homogeneous Neumann boundary conditions. Eigenfunctions of the Laplace operator on planar domains are well-studied in the literature \cite{paper9, paper11, paper17, thesis2}. It is presumed that $\bar{u}$ and $\bar{v}$ both are of a similar form. The eigenfunctions for the Laplace operator are found by solving the relevant eigenvalue problems, that satisfy the given boundary conditions of problem (\ref{B}). Such an eigenvalue problem is of the form
\begin{equation}    \label{eigen}
\begin{cases}
%\begin{equation}
    \frac{\partial^2 \bar{u}}{\partial x^2}+\frac{\partial^2 \bar{u}}{\partial y^2} = \eta \bar{u}, \qquad \eta \in \mathbb{R},
\\
\frac{\partial \bar{u}}{\partial x}(0,y) = 0,  \qquad \frac{\partial \bar{u}}{\partial x}(L,y) = 0,\qquad 0 \leq y \leq L,\\
\frac{\partial \bar{u}}{\partial y}(x,0) = 0,  \qquad \frac{\partial \bar{u}}{\partial y}(x,L) = 0,\qquad 0 \leq x \leq L.
\end{cases}
\end{equation}
The method of separation of variables is used to find the solution to problem (\ref{eigen}) by presuming that the solution is in the form of a product of two functions namely $X(x)$ and $Y(y)$, so $\bar{u}$ has the form
%\[
   $ \bar{u}(x,y) = X(x)Y(y)$.
%\]
This form of solution is substituted in equation (\ref{eigen}) to obtain two one-dimensional eigenvalue problems, which are individually solved and the resulting set of eigenfunctions solving (\ref{eigen}) are given  by
%\[
$\bar{u}_{n,m}(x,y) = C_{n,m}\cos{\big(\frac{n\pi x}{L}\big)}\cos{\big(\frac{m \pi y}{L}\big)}$,
%\]
where $C_{n,m}$ are the coefficients depending on the mode of the eigenfunctions.
The solution to problem (\ref{vect}) can be, therefore, written as the sum of infinite expansion in the form of the product of $T(t)$, and the eigenfunctions for the two dimensional Laplace operator, so we have
\[\begin{split}
\bar{u}(x,y,t) = \sum_{n=0}^\infty \sum_{m=0}^\infty U_{n,m} \exp(\lambda_{n,m} t) \cos{\big(\frac{n\pi x}{L}\big)}\cos{\big(\frac{m \pi y}{L}\big)},\\
\bar{v}(x,y,t) = \sum_{n=0}^\infty \sum_{m=0}^\infty V_{n,m}  \exp(\lambda_{n,m} t) \cos{\big(\frac{n\pi x}{L}\big)}\cos{\big(\frac{m \pi y}{L}\big)},
\end{split}\]
where $U_{n,m}$ and $V_{n,m}$ are the coefficients of the infinite expansion.
Substituting this form of solution in (\ref{vect}), the problem can be written as a two-dimensional discrete set of algebraic equations in the form, with $\lambda = \lambda_{n,m}$
\begin{equation}
\lambda \left[\begin{array}{c}
     \bar{u} \\
     \bar{v} 
\end{array}\right]=-\frac{(n^2+m^2) \pi^2}{L^2}\left[\begin{array}{cc}
     1&0  \\
     0&d 
\end{array}\right]\left[\begin{array}{c}
     \bar{u}\\
     \bar{v} 
\end{array}\right]+\gamma \left[\begin{array}{cc}
     \frac{\beta-\alpha}{\beta+\alpha}& (\beta+\alpha)^2  \\
     -\frac{2\beta}{\beta+\alpha}&-(\beta+\alpha)^2
\end{array}\right]\left[\begin{array}{c}
     \bar{u} \\
     \bar{v}
\end{array}\right],
\label{vect1}
\end{equation}
which can also be written as a two-dimensional discrete eigenvalue problem 
\begin{equation}
\left[\begin{array}{cc}
     \gamma \frac{\beta-\alpha}{\beta+\alpha}-\frac{(n^2+m^2) \pi^2}{L^2}& \gamma (\beta+\alpha)^2  \\
     -\gamma \frac{2\beta}{\beta+\alpha}&-\gamma(\beta+\alpha)^2-d\frac{(n^2+m^2) \pi^2}{L^2}
\end{array}\right]\left[\begin{array}{c}
     \bar{u} \\
     \bar{v}
\end{array}\right]=\lambda \left[\begin{array}{c}
     \bar{u} \\
     \bar{v} 
\end{array}\right].
\label{vect2}
\end{equation}
%\subsection{Characteristic polynomial}
For the stability of the steady state $(u_s,v_s)$ in the presence of diffusion, the eigenvalues of problem (\ref{vect2}) are analysed. The characteristic polynomial for (\ref{vect2}) is given by 
\begin{equation}
\left|\begin{array}{cc}
     \gamma \frac{\beta-\alpha}{\beta+\alpha}-\frac{(n^2+m^2) \pi^2}{L^2} -\lambda& \gamma (\beta+\alpha)^2  \\
     -\gamma \frac{2\beta}{\beta+\alpha}&-\gamma(\beta+\alpha)^2-d\frac{(n^2+m^2) \pi^2}{L^2}-\lambda
\end{array}\right|=0,
\label{vect4}
\end{equation}
which is a quadratic polynomial in $\lambda$. Let $\mathcal{T}$ and $\mathcal{D}$ denote the trace and determinant of the stability matrix (\ref{vect2}) then the pair of eigenvalues $\lambda_{1,2}$ of system (\ref{vect2}) satisfy the quadratic polynomial in terms of the trace and determinant in the form  
\begin{equation}
\lambda^2 - \mathcal{T}(\alpha, \beta) \lambda + \mathcal{D}(\alpha, \beta) = 0,
\label{charac}
\end{equation}
where $\mathcal{T}(\alpha, \beta)$ and $\mathcal{D}(\alpha, \beta)$  are given by 
\[\begin{split}
\mathcal{T}(\alpha, \beta) =& \gamma \frac{\beta - \alpha  - (\beta +\alpha)^3}{\beta +\alpha}-(d+1)\frac{(n^2+m^2) \pi^2}{L^2},\\
\mathcal{D}(\alpha, \beta) =& \Big(\gamma\frac{\beta-\alpha}{\beta+\alpha}-\frac{(n^2+m^2) \pi^2}{L^2}\Big)\Big(-\gamma(\beta +\alpha)^2-(d+1)\frac{(n^2+m^2) \pi^2}{L^2}\Big)+2\gamma^2\beta(\beta+\alpha).
\end{split}\]
The roots of (\ref{charac}) are
%\[
$\lambda_{1,2} = \frac{1}{2}\mathcal{T}(\alpha,\beta) \pm \frac{1}{2}\sqrt{\mathcal{T}^2(\alpha,\beta) -4\mathcal{D}(\alpha,\beta)}$.
%\]
We note that, in comparison to the case when no diffusion is present, the trace $\mathcal{T}(\alpha,\beta)$ of the Jacobian matrix  becomes more negative when diffusion is introduced since a positive quantity is being subtracted, thereby enhancing the stability of the system. This is equivalent to the default intuition that diffusion may only promote stability, due to its natural property of homogenising concentration gradient. However, the modification to determinant part $\mathcal{D}(\alpha,\beta)$ in the the form of an additional positive term $ (d+1)\pi^4\frac{(n^2+m^2)^2}{L^4}$ is intuitively expected to dominate the strictly positive term which is subtracted from  $\mathcal{T}(\alpha,\beta)$, namely $(d+1)\frac{(n^2+m^2) \pi^2}{L^2}$. These modifications to the trace and the determinant of the stability matrix due to the presence of diffusion to the system are investigated thoroughly to derive possible quantitative relationships between domain size, reaction and diffusion rates in the context of the parameter classification. We also attempt to quantitatively understand how this modification of both $\mathcal{T}(\alpha,\beta)$ and $\mathcal{D}(\alpha,\beta)$ contributes to creating instability. The regions of the parameter space for all types of steady states and their stability is analysed. Once the regions for different types of steady states are classified, then subject to this classification the influence of the domain size $\Omega$ is explored on the types and existence of diffusion-driven instabilities. 

\subsection{Analysis for the case of complex eigenvalues}
To determine the types of steady state $(u_s,v_s)$, with corresponding parameter space, it is important to realise that despite the dependence on $\alpha$ and $\beta$, the type of steady state also depends on the additional parameters namely $\gamma$, $d$ and $L$. The variation of the additional parameters are investigated for fixed $n$ and $m$ to explore how the parameter space varies for different values of $d$, for example. Another important relationship to investigate is between the area of the domain $L^2$ and the reaction rate $\gamma$ and how these influence the process of diffusion-driven instability. First, we analyse the case of complex eigenvalues. The eigenvalues $\lambda_{1,2}$ are complex values if $(\alpha,\beta)$ are chosen such that the discriminant of (\ref{charac}) is negative. It means that $(\alpha,\beta)$ must satisfy  
\begin{equation}
\frac{1}{4}\mathcal{T}^2(\alpha,\beta)<\mathcal{D}(\alpha,\beta).
\label{spiral}
\end{equation}
One immediate consequence of this inequality is that $\mathcal{D}(\alpha,\beta)$ must be positive, because if $\mathcal{D}(\alpha,\beta)$ is negative, then for no choice of $(\alpha,\beta)$ do the eigenvalues $\lambda_{1,2}$ become complex. This immediate observation is automatically in agreement with the conditions for diffusion-driven instability presented in \cite{paper9, paper17, paper20}. Furthermore, this condition on one hand will determine the region of the $(\alpha,\beta)\in\mathbb{R}_+$ plane on which $\lambda_{1,2}$ are a complex conjugate pair. On the other hand it serves to indicate the equations for the partitioning implicit curves on which the eigenvalues change type from real to complex. Bearing in mind, that inequality (\ref{spiral}) is a sufficient condition for $\lambda_{1,2}$ to have imaginary parts, therefore, it suffices to study (\ref{spiral}) for $(u_s,v_s)$ to be a spiral.
Before numerical treatment of the critical curve for the complex region, the stability of $(u_s,v_s)$ is studied analytically given that $\lambda_{1,2}$ have imaginary parts. Stability of $(u_s,v_s)$ is determined by the sign of the real part of $\lambda_{1,2}$ when it is a pair of complex conjugate values. Given that this is the case then the real part of $\lambda_{1,2}$ is given by
\[
\text{Re}(\lambda_{1,2})=\frac{1}{2}\Big(\gamma\frac{\beta - \alpha  - (\beta +\alpha)^3}{\beta +\alpha}-(d+1)\frac{(n^2+m^2) \pi^2}{L^2}\Big).
\]
If we wish to find the parameter space for which $(u_s,v_s)$ is a stable spiral then we require the real part of both eigenvalues to be negative, which means in addition to the condition (\ref{spiral}) we want $\mathcal{T}(\alpha,\beta)$ to be negative or equivalently $(\alpha,\beta)$ must satisfy
\begin{equation}
 \gamma\frac{\beta - \alpha  - (\beta +\alpha)^3}{\beta +\alpha}<(d+1)\frac{(n^2+m^2) \pi^2}{L^2}.
\label{stabspiral}
\end{equation}
In comparison to the absence of diffusion, the condition for stability of $(u_s,v_s)$ when $\lambda_{1,2}$ are complex eigenvalues is similar to (\ref{stabspiral}) except that the right-hand side of inequality is a strictly positive real quantity. The quantity on the right-hand side of (\ref{stabspiral}) has $L^2$ in the denominator. 
%This indicates inverse relationship of the condition for a spiral to become unstable with the area of $\Omega$. In presence of diffusion the right hand side is a strictly positive quantity.  
%The trace $T(\alpha,\beta)$ in the absence of diffusion provided some regions where it is positive, and hence the existence of regions where $(u_s,v_s)$ is an unstable spiral, which was shown in Figure \ref{fig5} (a).
% In the presence of diffusion a strictly positive term namely $(d+1)\frac{(n^2+m^2) \pi^2}{L^2}$ is subtracted, that is directly proportional to the extra parameters $d$ and inversely proportional to the area of $\Omega$ and reaction rate $\gamma$, which are also positive quantities. 
 The maximum value that $T(\alpha,\beta)$ in the absence of diffusion could attain on $(\alpha,\beta)\in \mathbb{R}_+$ is 1. The expression for $\mathcal{T}(\alpha,\beta)$ in the presence of diffusion is the same as $T(\alpha,\beta)$ in the absence of diffusion except it is reduced by a strictly positive constant and multiplied by $\gamma$. This relationship indicates that the stability of $(u_s,v_s)$ when $\lambda_{1,2}$ are complex, depends on the extra parameters $L$, $d$ and $\gamma$. The maximum of the left-hand side of (\ref{stabspiral}) is the non-zero positive constant $\gamma$ [see Section 4, Theorem \ref{theorem1}], which suggests that the real part of $\lambda_{1,2}$ is negative if the term on the right-hand side of (\ref{stabspiral}) satisfies
\begin{equation}
(d+1)\frac{(n^2+m^2) \pi^2}{L^2}>\gamma \qquad \iff \qquad L^2 <\pi^2\frac{(d+1)(m^2+n^2)}{\gamma}.
\label{cond}
\end{equation}
In the presence of diffusion, if $\lambda_{1,2}$ is a pair of complex conjugate eigenvalues then the real part of $\lambda_{1,2}$ is always negative given that the area $L^2$ of $\Omega$ satisfies (\ref{cond}), otherwise no choice of $\alpha>0$, $\beta>0$ and $d>0$ can change the sign of $\text{Re}(\lambda_{1,2})$ to positive. This claim is formally proven in Theorem \ref{theorem1} as well as numerically demonstrated that no choice of $\alpha$, $\beta$ and $d$ can invalidate inequality (\ref{stabspiral}), given that $L$ satisfies (\ref{cond}).
For the first eigenvalue of the laplace operator, which corresponds to $m=n=0$, inequality (\ref{spiral}) takes exactly the same form as the inequality that was necessary for $(u_s,v_s)$ to be a spiral in the absence of diffusion, namely inequality (\ref{fiq}). Therefore, to avoid repetition the stability analysis of the steady state $(u_s,v_s)$ with $m=n=0$ is skipped and the non-trivial case is studied, i.e.  $m,n\geq1$. 
In the absence of diffusion, in the region where $(u_s,v_s)$ is a spiral, there are  sub-regions that correspond to a stable spiral as well as sub-regions that correspond to an unstable spiral. One immediate consequence of including diffusion, into the model, is that it puts the condition (\ref{cond}) on the domain size $L^2$ for the real part of $\lambda_{1,2}$ to remain negative. Similarly the real part of $\lambda_{1,2}$ becomes positive if condition (\ref{cond}) is violated. This means, diffusion-driven instability only serves to destabilise a spiral type of uniform steady state $(u_s,v_s)$ if the parameter $L$ violates inequality (\ref{cond}), which means 
\begin{equation}
L^2 >\pi^2\frac{(d+1)(m^2+n^2)}{\gamma}.
\label{cond1}
\end{equation}
Conditions (\ref{cond}) and (\ref{cond1}) both have consequences in terms of the stability of the original system. The first consequence is that if $\lambda_{1,2}$ have imaginary parts, then condition (\ref{cond}) guarantees that the real part of $\lambda_{1,2}$ is negative for all choices of $\alpha$ and $\beta$, allowing the eigenvalues to be positive if and only if they are both real repeated or real distinct (with at least one of them positive). Hence, restricting the system from undergoing \textit{Hopf} bifurcation. Similarly, if $\lambda_{1,2}$ have imaginary parts, then condition (\ref{cond}) must be violated (condition (\ref{cond1}) must be satisfied) in order for the real part of $\lambda_{1,2}$ to become positive whilst the eigenvalues are still complex conjugate pair. It means that for the original system to undergo \textit{Hopf} bifurcation, the necessary condition on the domain size is (\ref{cond1}). Both of these statements are formally proven in the form of theorems in the Section \ref{main}.
If the real part of the complex eigenvalues becomes zero, then the roots become purely imaginary, indicating that the system undergoes a time-periodic oscillations whilst experiencing spatial stability. This type of instability is referred to as \textit{Transcritical bifurcation} and this entails that the real part of $\lambda_{1,2}$ must become zero, which means
\begin{equation}
   \gamma\frac{\beta - \alpha  - (\beta +\alpha)^3}{\beta +\alpha}-(d+1)\frac{(n^2+m^2) \pi^2}{L^2}=0.
    \label{trans}
\end{equation}
Rearranging (\ref{trans}) and solving for $L^2$, the necessary condition for transcritical bifurcation is obtained in the form
\begin{equation}
    L^2 = \frac{(\beta+\alpha)(d+1)(n^2+m^2)\pi^2}{\gamma\big(\beta-\alpha-(\alpha+\beta)^3\big)}.
    \label{cond2}
\end{equation}
If the parameter $L$ satisfies condition (\ref{cond1}) simultaneously with condition (\ref{cond2}) and $\mathcal{D}(\alpha,\beta)>0$, then the system exhibits a transcritical bifurcation.

\subsection{Interpretation of the dynamics for the case of complex eigenvalues}
For a domain size satisfying (\ref{cond}) the dynamics of the system with diffusion is expected to undergo spatial patterning only. If the diffusion rate $d$ is increased, yet satisfying (\ref{cond}), instability is still expected to invade stable regions in the parameter plane, however the only type of instability one may expect is spatial and not temporal. This type of instability is referred to as \textit{Turing Instability} \cite{book1, paper4, paper26}, where a certain region in the parameter space is stable in the absence of diffusion and yet becomes spatially unstable upon adding diffusion to the system. Findings of the current study suggest that, if the relationship between $L$, $\gamma$ and $d$ is maintained as given by (\ref{cond}), the type of diffusion-driven instability is always restricted to \textit{Turing Instability}. It means that the only pattern one can achieve with (\ref{cond}) is spatially periodic pattern, with no temporal periodicity. However, if $L$, $\gamma$ and $d$ are given to satisfy condition (\ref{cond1}) then the system may undergo both temporal and/or spatial periodicity. Upon increasing the value of $d$, yet maintaining condition (\ref{cond1}) the system is expected to become temporally unstable. This type of instability is referred to as \textit{Hopf Bifurcation}, where adding diffusion to a temporally stable system, causes temporal instability. In this case one may expect the system to undergo temporal periodicity (pattern along the \textit{time-axis}). 
 \subsection{Analysis for the case of real eigenvalues}
The eigenvalues $\lambda_{1,2}$ are both real if the discriminant of the roots is either zero or positive, which in turn means that both eigenvalues are real values if the relationship between $\mathcal{T}(\alpha,\beta)$ and $\mathcal{D}(\alpha,\beta)$ is such that 
\begin{equation}
\mathcal{T}^2(\alpha,\beta) \geq 4\mathcal{D}(\alpha,\beta).
\label{condreal}
\end{equation}
The equal case of (\ref{condreal}) is looked at first, where we have
\begin{equation}
\mathcal{T}^2(\alpha,\beta) = 4\mathcal{D}(\alpha,\beta),
\label{condequal}
\end{equation}
which means that the discriminant is zero, hence the roots are repeated real values of the form $\lambda_1 = \lambda_2\in\mathbb{R}$, given by
\begin{equation}
\lambda_1=\lambda_2 = \frac{1}{2}\Big(\gamma\frac{\beta - \alpha  - (\beta +\alpha)^3}{\beta +\alpha}-(d+1)\frac{(n^2+m^2) \pi^2}{L^2}\Big).
\label{repeated}
\end{equation}
When $\alpha$ and $\beta$ satisfy condition (\ref{condequal}), the stability of the steady state is determined by the sign of the root itself. 
The expression given by (\ref{repeated}) can be easily shown to be negative if the area $L^2$ of the domain  satisfies the inequality
\begin{equation}
L^2<\pi^2\frac{(\beta+\alpha)(d+1)(n^2+m^2)}{\gamma\big(\beta-\alpha-(\beta+\alpha)^3\big)}.
\label{repeatneg}
\end{equation}
 Otherwise, the repeated root is positive provided that $L$ satisfies
\begin{equation}
L^2>\pi^2\frac{(\beta+\alpha)(d+1)(n^2+m^2)}{\gamma\big(\beta-\alpha-(\beta+\alpha)^3\big)}.
\label{repeatpos}
\end{equation}
Analysing (\ref{repeatneg}) and (\ref{repeatpos}) carefully, it can be observed that the only terms that can possibly invalidate the inequalities are in the denominator of the right hand-side, namely the expression $\beta-\alpha-(\beta+\alpha)^3$. Therefore, a restriction is required to be stated on this term to ensure that the area of $\Omega$ is not compared against a negative quantity, such a restriction is 
\begin{equation}
\beta > \alpha+(\beta+\alpha)^3.
\label{rest}
\end{equation}
%\subsubsection{Remark}
It must be noted that (\ref{rest}) is the same restriction on the parameter choice obtained for the case of repeated real eigenvalues in the absence of diffusion. The region where the eigenvalues are real repeated roots are implicit curves in the parameter space satisfying (\ref{condreal}), these curves are computed numerically in the last part of Section 4. These curves form the boundary between the regions of complex and real eigenvalues. Varying the diffusion rate $d$ causes a shift to the location of the curves indicating clearly regions that are subject to diffusion-driven instability.
The remaining case to look at is when both eigenvalues are real distinct. This happens if $\alpha$ and $\beta$ are chosen such that the strict inequality case of (\ref{condreal}) is satisfied. This case corresponds to the diffusion-driven instability \textit{Turing type} only, because both eigenvalues are real and distinct. 
\subsection{Interpretation of the dynamics for the case of real eigenvalues}
If both eigenvalues are negative distinct real values, then the system is spatially as well as temporally stable, the dynamics will achieve no patterns, hence the system returns to the uniform constant steady state $(u_s,v_s)$ as time grows, [see Section 5 Figure \ref{fig14}] with no effect from diffusion. If the eigenvalues are both real with different signs, then the type of instability caused by diffusion is spatially periodic or oscillatory in space, because this case corresponds to the steady state becoming a saddle point. If both eigenvalues are positive real distinct values, then the dynamics are expected to exhibit a spatially periodic pattern, in the form of stripes or spots. 
\section{Main results}\label{main}
This section contains the main results of our current study. The analytical findings are formally presented in the form of two theorems with proofs. Each of which indicates the relationship between the domain size with the diffusion coefficient $d$ and reaction rate $\gamma$. The numerical technique for computing the critical boundary curves is also briefly explained and results of parameter space classification are numerically computed. Numerical verification of both theorems is carried out in this section, where the reaction-diffusion system is numerically solved using the finite element method on a unit square domain. The relationship between domain length and parameters $d$ and $\gamma$ is shown to be in agreement with theoretical predictions presented.

\begin{theorem}[Turing type diffusion-driven instability] Let $u$ and $v$ satisfy the non-dimensional reaction-diffusion system with {\it activator-depleted} reaction kinetics \eqref{B}-\eqref{zerofluxBCs} on a square domain $\Omega \subset \mathbb{R}^2$ with area $L^2$ and positive parameters $\gamma>0$, $d>0$, $\alpha>0$ and $\beta>0$.
% in the form 
% \begin{equation}\begin{split}
%   \frac{\partial u}{\partial t} = \triangle u + \gamma( \alpha -u+u^2v )\\
%   \frac{\partial v}{\partial t} = d \triangle v + \gamma(\beta - u^2v).
% \label{sys}\end{split}\end{equation}
 Given that the area of the square domain $\Omega \subset \mathbb{R}^2$ satisfies the inequality \eqref{cond}
%  \begin{equation}
% L^2<\pi^2\frac{(d+1)(m^2+n^2)}{\gamma},
% \label{cond}
% \end{equation}
 where $m,n \in \mathbb{N}$ then for all $\alpha, \beta \in \mathbb{R}_+$ in the neighbourhood of the unique steady state $(u_s,v_s)=\big(\alpha+\beta, \frac{\beta}{(\alpha+\beta)^2}\big)$ the diffusion driven instability is restricted to Turing type only, forbidding the existence of Hopf and transcritical bifurcation.
 \label{theorem1}
 \end{theorem}
 %\proof
 \begin{proof}
 The strategy of this proof is through detailed analysis of the real part of the eigenvalues of the linearised system, when the eigenvalues are a complex conjugate pair. This can be done through studying the surface $\mathcal{T}(\alpha,\beta)$, and finding that it has a unique extremum point at $(0,0)$.  The method of second derivative test and Hessian matrix is used to determine the type of this extremum. Upon finding its type, then the monotonicity of $\mathcal{T}(\alpha,\beta)$ is analysed in the neighbourhood of the extremum point in both directions $\alpha$ and $\beta$. The monotonicity analysis and the type of the extremum leads to proving the claim of the theorem.
 
The eigenvalues $\lambda_{1,2}$ in the presence of diffusion, in terms of trace $\mathcal{T}(\alpha,\beta)$ and determinant $\mathcal{D}(\alpha,\beta)$ are given by
%\[
 $\lambda_{1,2} = \frac{1}{2}\mathcal{T}(\alpha,\beta) \pm \frac{1}{2}\sqrt{\mathcal{T}^2(\alpha,\beta) -4\mathcal{D}(\alpha,\beta)}$,
 %\]
 where 
 \[\begin{split}
\mathcal{T}(\alpha,\beta)=& \gamma \frac{\beta - \alpha  - (\beta +\alpha)^3}{\beta +\alpha}-(d+1)\frac{(n^2+m^2) \pi^2}{L^2},\\
 \mathcal{D}(\alpha,\beta) =&\Big(\gamma \frac{\beta-\alpha}{\beta+\alpha}-\frac{(n^2+m^2) \pi^2}{L^2}\Big)\Big(-\gamma(\beta +\alpha)^2-d\frac{(n^2+m^2) \pi^2}{L^2}\Big)+2\gamma^2\beta(\beta+\alpha).
 \end{split}\]
 It can be immediately observed that in the neighbourhood of $(u_s,v_s)$ for the system to exhibit Hopf or transcritical bifurcation the discriminant of the characteristic polynomial must satisfy the inequality
% \[
 $\mathcal{T}^2(\alpha,\beta) -4\mathcal{D}(\alpha,\beta)<0$.
% \]
 Therefore, the stability and type of the steady state $(u_s,v_s)$ in this case is determined by the sign of the real part of $\lambda_{1,2}$. The aim is to investigate $\mathcal{T}(\alpha,\beta)$ and derive from it condition (\ref{cond}) on $L^2$ as a requirement for $\mathcal{T}(\alpha,\beta)$ to be negative for all strictly positive choices of $\gamma$, $\alpha$, $\beta$ and $d>0$.
 First derivative test is used on $\mathcal{T}(\alpha,\beta)$ to find all stationary points of $\mathcal{T}(\alpha,\beta)$ on the domain $[0,\infty) \times [0,\infty)$.
 All stationary points of $\mathcal{T}(\alpha,\beta)$ must satisfy
 %\[
 $\frac{\partial \mathcal{T}}{\partial \alpha} = - \gamma\frac{2(\alpha+\beta)^3+2\beta}{(\alpha+\beta)^2} = 0$,
 %\]
 which is true if and only if 
 \begin{equation}
  (\alpha +\beta)^3+\beta = 0.
  \label{first}
 \end{equation}
 Similarly all stationary points of $\mathcal{T}(\alpha,\beta)$ must also satisfy
% \[
 $\frac{\partial \mathcal{T}}{\partial \beta} = - \gamma \frac{2(\alpha+\beta)^3-2\alpha}{(\alpha+\beta)^2} = 0$,
% \]
 which implies 
  \begin{equation}
  (\alpha +\beta)^3-\alpha = 0.
  \label{second}
 \end{equation}
 The system of nonlinear algebraic equations obtained from (\ref{first}) and (\ref{second}) has a unique solution namely $\alpha=0$ and $\beta = 0$ [see Remark \ref{remarkone}].
 Therefore, $\mathcal{T}(\alpha,\beta)$ has a unique stationary point at the origin. The type of this stationary point is determined by the second derivative test for which the Hessian matrix $H(\mathcal{T}(\alpha,\beta))$ must be computed and evaluated at the point $(0,0)$.
 
 \begin{equation}
 \begin{split}
H(\mathcal{T}(\alpha,\beta))|_{(0,0)} =\left[\begin{array}{cc}
   \frac{\partial ^2 \mathcal{T}}{\partial \alpha^2}  &  \frac{\partial ^2 \mathcal{T}}{\partial \beta\partial \alpha } \\
      \frac{\partial ^2 \mathcal{T}}{\partial \alpha \partial \beta}&\frac{\partial ^2 \mathcal{T}}{\partial \beta^2}
\end{array}\right]_{(0,0)} = \left[\begin{array}{cc}
   -\gamma \frac{4\beta-2(\alpha+\beta)^3}{(\alpha+\beta)^3}  &  -\gamma \frac{2(\alpha+\beta)^3+2(\alpha-\beta)}{(\alpha+\beta)^3} \\
      -\gamma\frac{2(\alpha+\beta)^3+2(\alpha-\beta)}{(\alpha+\beta)^3}&-\gamma\frac{2(\alpha+\beta)^3+4\alpha}{(\alpha+\beta)^3}
\end{array}\right]_{(0,0)}.
\notag
\end{split}
\end{equation}
It is clear that the entries of $H$ upon direct evaluation at the point $(0,0)$ are undefined. This is treated by using L'Hopital's rule.  L'Hopital's rule sometimes does not work for functions of two variables defined on cartesian coordinates, therefore a transformation of the entries to polar coordinates might be applied. We will exploit this technique to express the Hessian matrix in polar coordinates and differentiate accordingly. 
%The main reason why this transformation works for this particular case is because, it in a sense combines both variables $\alpha$ and $\beta$ into a single variable $r$, which denotes the radial distance from the point $(0,0)$. Since in the infinitesimal vicinity of the origin $(0,0)$ the variation of the second variable of polar coordinates $\theta$ can be disregarded. The limit $\lim_{(\alpha,\beta)\rightarrow (0,0)}\mathcal{T}(\alpha,\beta)$ for each entry can be computed as $\lim_{r\rightarrow 0} \mathcal{T}(r,\theta)$, in which differentiation with respect to single variable is required. 
The entries of $H$ are transformed to polar coordinates using $\alpha = r\cos(\theta)$ and $\beta = r\sin(\theta)$, so the rule can be applied by taking the $\lim_{r\rightarrow 0} H$.
Using $(r,\theta)$ coordinates the entries of $H$ take the following form
\begin{equation}
H(\mathcal{T}(r,\theta))|_{r = 0}=-\gamma \left[\begin{array}{cc}
   \frac{4r\sin\theta-2r^3(\cos\theta+\sin\theta)^3}{r^3(\cos\theta+\sin\theta)^3}  &  \frac{2r^3(\cos\theta+\sin\theta)^3+2r(\cos\theta-\sin\theta)}{r^3(\cos\theta+\sin\theta)^3}\\
      \frac{2r^3(\cos\theta+\sin\theta)^3+2r(\cos\theta-\sin\theta)}{r^3(\cos\theta+\sin\theta)^3}&\frac{4r\cos\theta+2r^3(\cos\theta+\sin\theta)^3}{r^3(\cos\theta+\sin\theta)^3} 
\end{array}\right]_{r=0}.
\label{hess1}
\end{equation}
L'Hopital's rule is applied to each entry of $H$ separately and the $\lim_{r \rightarrow 0} H_{ij}(\mathcal{T}(r,\theta))$ is computed for $i,j=1,2$. Starting with the entry $H_{11}$ and cancelling $r$, the expression takes the form
 \[
 \begin{split}
 \lim_{r \rightarrow 0} H_{11} =\lim_{r \rightarrow 0}\frac{4\sin\theta-2r^2(\cos\theta+\sin\theta)^3}{r^2(\cos\theta+\sin\theta)^3}.
 \end{split}
 \]
 Let $\mathcal{T}_1(r,\theta)$ and $\mathcal{T}_2(r,\theta)$ respectively denote the numerator and the denominator of the expression for $H_{11}$, then the application of L'Hopital's rule suggests that 
 \[\begin{split}
 \lim_{r \rightarrow 0} H_{11}(\mathcal{T}(r,\theta)) &= \lim_{r \rightarrow 0} \frac{\mathcal{T}_1(r,\theta)}{\mathcal{T}_2(r,\theta)}= \frac{\lim_{r \rightarrow 0} \frac{d\mathcal{T}_1}{dr}(r,\theta)}{\lim_{r \rightarrow 0} \frac{d\mathcal{T}_2}{dr}(r,\theta)}= \lim_{r \rightarrow 0}\frac{ -4r(\cos \theta + \sin \theta)^3}{2 r (\cos \theta +\sin \theta)^3}= -2.
 \end{split}\]
 Applying the same procedure for $H_{12}$, $H_{21}$ and $H_{22}$, all the entries of $H$ are computed and given by
 \begin{equation}
H(\mathcal{T}(\alpha,\beta))|_{(0,0)}=-\gamma\left[\begin{array}{cc}
   -2  &  2 \\
  2  &  2
\end{array}\right].
\label{hessnum}
\end{equation}
Since the $\text{det}(H) = -8 \gamma^2< 0$, therefore, the second derivative test suggests that $(0,0)$ is a saddle point of $\mathcal{T}(\alpha,\beta)$. Since it was previously shown that $\mathcal{T}(\alpha,\beta)$ attains a unique stationary point in the domain $[0,\infty)\times[0,\infty)$, i.e. by solving the equations (\ref{first}) and (\ref{second}), therefore, if $(0,0)$ was a maximum and $\mathcal{T}(0,0)<0$, this would suggest that, whenever $\lambda_{1,2}$ has a non-zero imaginary part then $Re(\lambda_{1,2})<0$ regardless of the choice of $d$, $\gamma$ and $L^2$, however due to fact that $(0,0)$ is a saddle point, it remains to show that $\mathcal{T}(\alpha,\beta)$ is negative at $(0,0)$ and its first derivatives in the neighbourhood of $(0,0)$ of $\mathcal{T}(\alpha,0)$ and $\mathcal{T}(0,\beta)$ for positive values of $\alpha$ and $\beta$ in both directions are negative and do not change sign. Let $\mathcal{T}_0(\alpha)$ and $\mathcal{T}_0(\beta)$ denote the curves for constants $\beta=0$ and $\alpha=0$ respectively on the plane $\mathcal{T}(\alpha,\beta)$, then 
\[\begin{split}
\mathcal{T}_0(\alpha) = \lim_{\beta \rightarrow 0}\mathcal{T}(\alpha,\beta) = -\gamma(1+\alpha^2)-(d+1)\frac{(m^2+n^2)\pi^2}{L^2},\\
\mathcal{T}_0(\beta) = \lim_{\alpha \rightarrow 0}\mathcal{T}(\alpha,\beta) = \gamma(1-\beta^2)-(d+1)\frac{(m^2+n^2)\pi^2}{L^2}.
\end{split}\]
The expression for $\mathcal{T}_0(\alpha)$ clearly satisfy that it is negative at $\alpha=0$ and its first derivative in the direction of $\alpha$ is
%\[
$ \frac{d\mathcal{T}_0(\alpha)}{d\alpha} = -2\gamma\alpha <0$ for all  $\gamma, \alpha \in [0,\infty)$.
%\]
The expression for $\mathcal{T}_0(\beta)$ however is not trivially negative for all values, since the sign of the constant $\gamma$ in the expression is positive, which if computed at $\beta=0$, leads to the desired condition
\[
\mathcal{T}_0(\beta)\big|_{\beta=0} = \gamma-(d+1)\frac{(m^2+n^2)\pi^2}{L^2}<0 \qquad \iff \qquad L^2<\pi^2\frac{(d+1)(m^2+n^2)}{\gamma}.
\]
It has been shown that the condition (\ref{cond}) is necessary for $\mathcal{T}(\alpha,\beta)$ to be negative at the unique stationary point namely $(0,0)$, it remains to show that the first derivative of $\frac{d\mathcal{T}_0}{d\beta}(\beta)<0$,
%\[
$\frac{d\mathcal{T}_0}{d\beta} = -2\gamma \beta <0$  for all $ \gamma \beta \in [0,\infty)$
%\]
which completes the proof.
 \end{proof}
 
\subsubsection{Remark}
\label{remarkone}
  The parameters $(\alpha,\beta)=(0,0)$ are not the admissible choices for the original system, because this choice of parameters leads to the trivial steady state $(u_s,v_s)=(0,0)$, which is globally and unconditionally stable and has no physical relevance. However, the surface defined by $\mathcal{T}(\alpha,\beta)$, which determines the sign of the real part of $\lambda_{1,2}\in \mathbb{C}\backslash \mathbb{R}$ attains a unique extremum at $(0,0)$, which makes the analysis of the neighbourhood of this point important for studying the sign of $\mathcal{T}(\alpha,\beta)$. It is reasonable to argue that the point $(\alpha,\beta)=(0,0)$ is not permissible to use as the platform of the proof of Theorem \ref{theorem1}, in that case, proving the results become a step closer by only showing that $\mathcal{T}(\alpha,\beta)$ is a strictly monotonically decreasing function (with no extrema in $(\alpha,\beta)\in \mathbb{R}^2_+$) and it can only attain bounded positive values in the neighbourhood of $(0,0)$ if and only if the condition (\ref{cond}) on the domain size $L^2$ is violated. Therefore, given that $L^2$ maintains to satisfy (\ref{cond}), the sign of the real part i.e. $\mathcal{T}(\alpha,\beta)$ of $\lambda_{1,2}\in \mathbb{C}\backslash \mathbb{R}$ is guaranteed to be negative, which is a step shorter to reach the claim of Theorem \ref{theorem1}. Therefore, it is brought to the attention of the reader that the use of the point $(\alpha,\beta)=(0,0)$ as an extremum of $\mathcal{T}(\alpha,\beta)$ is more of complementary factor to the proof rather than an essential one.   

\begin{theorem}[Hopf or transcritical bifurcation]
 Let $u$ and $v$ satisfy the non-dimensional reaction-diffusion system with {\it activator-depleted} reaction kinetics \eqref{B}-\eqref{zerofluxBCs} on a square domain $\Omega \subset \mathbb{R}^2$ with length $L$ and positive parameters $\gamma$, $d>0$, $\alpha>0$ and $\beta>0$.
% , in the form 
% \[\begin{split}
%   \frac{\partial u}{\partial t} = \triangle u +\gamma( \alpha -u+u^2v )\\
%   \frac{\partial v}{\partial t} = d \triangle v + \gamma(\beta - u^2v).
% \end{split}\]
 For the system to exhibit Hopf or transcritical bifurcation in the neighbourhood of the unique steady state $(u_s,v_s)=\big(\alpha+\beta, \frac{\beta}{(\alpha+\beta)^2}\big)$, the necessary condition on the area $L^2$ of the square domain $\Omega \subset \mathbb{R}^2$ is \eqref{cond1}
% \begin{equation}
% L^2\geq\pi^2\frac{(d+1)(m^2+n^2)}{\gamma},
% \label{cond1}
% \end{equation} 
 where $m$ and $n$ are any integers.
 \label{theorem2}
 \end{theorem}
 \begin{proof}
 %\proof
The strategy to this proof is some what different from that of Theorem \ref{theorem1}, because in Theorem \ref{theorem1} the aim was to show that condition (\ref{cond}) must be satisfied by $L^2$ for $\mathcal{T}(\alpha,\beta)$ to be negative and to maintain negative sign throughout the $(\alpha,\beta)$ plane. For this proof the only important step is to derive the condition for $\mathcal{T}(\alpha,\beta)$ to be positive, because when $\lambda_{1,2}$ are complex eigenvalues then the sign of $\mathcal{T}(\alpha,\beta)$ is precisely what determines the stability of $(u_s,v_s)$. The system undergoes Hopf or transcritical bifurcation in the neighbourhood of $(u_s,v_s)$, if the sign $\mathcal{T}(\alpha,\beta)$ is positive. Therefore for diffusion-driven instability to influence $(u_s,v_s)$ when $\lambda_{1,2}$ are complex conjugate pair it is necessary that
\begin{equation}
\mathcal{T}(\alpha,\beta) =\gamma\frac{\beta - \alpha  - (\beta +\alpha)^3}{\beta +\alpha}-(d+1)\frac{(n^2+m^2) \pi^2}{L^2}\geq0,
\label{statement}
\end{equation}
which can equivalently be written as
\begin{equation}
\gamma\frac{\beta - \alpha  - (\beta +\alpha)^3}{\beta +\alpha}\geq(d+1)\frac{(n^2+m^2) \pi^2}{L^2}.
\label{t2}
\end{equation}
The expression on the left hand-side of (\ref{t2}) is explored in particular to find its upper bound, it can be written it in the form of difference between two rational functions as 
\begin{equation}\begin{split}
\gamma\frac{\beta - \alpha  - (\beta +\alpha)^3}{\beta +\alpha} = \gamma \Big(f_1(\alpha,\beta)-f_2(\alpha,\beta)\Big),
\label{uppbound}
\end{split}
\end{equation}
where $f_1(\alpha,\beta) = \frac{\beta}{\beta+\alpha}$  and $f_2(\alpha,\beta) = \frac{\alpha+(\beta+\alpha)^3}{\beta+\alpha}$. 
The range for $f_1(\alpha,\beta)$ and $f_2(\alpha,\beta)$ are independently analysed to find the supremum of the expression on the left of (\ref{uppbound}). Starting with $f_1(\alpha,\beta)$, which is bounded below and above in the domain $(\alpha,\beta)\in[0,\infty)\times [0,\infty)$, we have
$\sup_{\alpha,\beta \in \mathbb{R}_+} f_1(\alpha,\beta) = 1$, and 
the $\inf_{\alpha,\beta \in \mathbb{R}_+} f_1(\alpha,\beta) = 0$ for all $ \alpha,\beta \in \mathbb{R}_+.$
Similarly considering the expression for $f_2(\alpha,\beta)$, we have
$\sup_{\alpha,\beta \in \mathbb{R}_+} f_2(\alpha,\beta) = \infty, $ and 
the $\inf_{\alpha,\beta \in \mathbb{R}_+} f_2(\alpha,\beta) = 0$, for all $\alpha, \beta \in \mathbb{R}_+.$
Since the ranges of both $f_1(\alpha,\beta)$ and $f_2(\alpha,\beta)$ are non-negative within their respective domains, therefore the supremum of their difference is determined by the supremum of the function with positive sign, which is $\sup_{\alpha,\beta\in\mathbb{R}_+}f_1(\alpha,\beta)=1$. Therefore, inequality (\ref{statement}) can be written as
\[\begin{split}
(d+1)\frac{(n^2+m^2) \pi^2}{L^2}&\leq\gamma\frac{\beta - \alpha  - (\beta +\alpha)^3}{\beta +\alpha} \\
& \leq \gamma \sup_{\alpha,\beta\in \mathbb{R}_+} \big(f_1(\alpha,\beta)-f_2(\alpha,\beta)\big) = \gamma \sup_{\alpha,\beta \in \mathbb{R}_+} f_1(\alpha,\beta)  = \gamma,
\end{split}\]
which by rearranging leads to the desired condition \eqref{cond1}.
\end{proof}
%\endproof

\subsection{Equations of the implicit partitioning curves}
A similar method to that used in the case of the absence of diffusion is applied to solve for the partitioning curves in  the parameter plane. The partitioning curve for the complex region, on the plane $(\alpha,\beta)\in\mathbb{R}_+$ must satisfy
\begin{equation}
  \mathcal{T}^2(\alpha,\beta)-4\mathcal{D}(\alpha,\beta) =0,
  \label{crit}
\end{equation}
since this is the curve on which the discriminant of the expression for eigenvalues change sign. It means the curve satisfying (\ref{crit}) determines the boundary on one side of which $\lambda_{1,2}$ are both real and on the other side of it, $\lambda_{1,2}$ are both complex conjugate eigenvalues.
Equation (\ref{crit}) is solved by a similar method to that used for solving (\ref{imp}) on a square domain and for each fixed mesh point in the direction of $\alpha$ it is found that its solution is equivalent to finding the positive real roots of the polynomial of degree 6 in $\beta$ for fixed values of $\alpha_i$ in the form
\begin{equation}
\begin{split}
\psi(\alpha_i,\beta) &= C_0(\alpha_i)+C_1(\alpha_i)\beta+C_2(\alpha_i)\beta^2+C_3(\alpha_i)\beta^3\\&+C_4(\alpha_i)\beta^4+C_5(\alpha_i)\beta^5+C_6(\alpha_i)\beta^6,
\label{star2}
\end{split}
\end{equation}
where for each $\alpha_i$ the coefficients are given by
\begin{equation*}
\begin{split}
C_0(\alpha_i) =& {L}^{4}\gamma^2{\alpha_i}^{6}+\big(4\,{L}^{4}{\pi }^{2}{m}^{2}+4\,{L}^{4}
{\pi }^{2}{n}^{2}+2\,{L}^{2}\gamma{\pi }^{2}d{m}^{2}+2\,{L}^{2}\gamma{\pi }^{2}d{n}^{2}+4\,{L}^{6}\\
&+2\,{L}^{2}\gamma{\pi }^{2}{m}^{2}+2\,{L}^{2}\gamma{\pi }^{2}{n}^{2}+2\,{L}^{4}\gamma^2\big){\alpha_i}^4
+\big(-4\,{L}^{2}{\pi }^{4}d{m}^{4}
-8\,{L}^{2}{\pi }^{4}d{m}^{2}{n}^{2}\\&-4\,{L}^{2}{\pi }^{4}d{n
}^{4}-4\,{L}^{2}{\pi }^{4}{m}^{4}-8\,{L}^{2}{
\pi }^{4}{m}^{2}{n}^{2}-4\,{L}^{2}{\pi }^{4}{n}^{4}-4\,{L}^{4}{\pi }^{2}d{m}^{2}\\
&-4\,{L}^{4}{\pi }^{2}d{n
}^{2}-4\,{L}^{4}{\pi }^{2}{m}^{2}-4\,{L}^{4}{
\pi }^{2}{n}^{2}\big){\alpha_i}^{3}
 +\big({\pi }^{4}{d}^{2}{m}^{4}+2\,{
\pi }^{4}{d}^{2}{m}^{2}{n}^{2}\\&+{\pi }^{4}{d}^{2}{n}^{4}+2\,{\pi }^{4}d{m}^{4}+4\,{\pi }^{4}d{m}^{2}{n}
^{2}+2\,{\pi }^{4}d{n}^{4}+{\pi }^{4}{m}^{4}
+2\,{\pi }^{4}{m}^{2}{n}^{2}
\\&+{\pi }^{4}{n}^{4}
+2\,{L}^{2}\gamma{\pi }^{2}d{m}^{2}+2\,{L}^{2}\gamma{\pi }^
{2}d{n}^{2}+2\,{L}^{2}\gamma{\pi }^{2}{m}^{2}+2\,{L}
^{2}\gamma{\pi }^{2}{n}^{2}+{L}^{4}\gamma^2\big){\alpha_i}^{2},
\end{split}
\end{equation*}
\begin{equation*}
\begin{split}
C_1(\alpha_i) =& 6\,{L}^{4}\gamma^2{\alpha_i}^{5}-8\,{L}^{4}{\alpha_i}^{4}+\big(8\,{L}^{4}{\pi }^{2}{m}^{2}+8\,{L}^{4}{\pi }^{2}{n}^{2}+8\,{L}^{2}\gamma{\pi }^{2}d{m}^{2}+8\,{L}^{2}\gamma{\pi }^{2}d{n}^{2}\\&
+8\,{L}^{2}\gamma{\pi }^{2}{m}^{2}+8\,{L}^{2}\gamma{\pi }^{2}{n}^{2}
+4\,{L}^{4}\gamma^2\big){\alpha_i}^{3}
+\big(-12\,{L}^{2}{\pi }^{4}d{m}^{4}-24\,{L}^{2}{\pi }^{4}d{m}^{2}{n}^{2}\\&-12\,{L}^{
2}{\pi }^{4}d{n}^{4}
-12\,{L}^{2}{\pi }^{4}{m}^{4}-24\,{L}^{2}{\pi }^{4}{m}^{2}{n}^{2}
-12\,{L}^{2}{\pi }^
{4}{n}^{4}-4\,{L}^{4}{\pi }^{2}d{m}^{2}\\&-4\,{L}
^{4}{\pi }^{2}d{n}^{2}-4\,{L}^{4}{\pi }^{2}{m}^{2}-4\,{L}^{4}{\pi }^{2}{n}^{2}\big){\alpha_i}^{2}
+\big(2\,{\pi }^{4}{d}^{2}{m}^{
4}+4\,{\pi }^{4}{d}^{2}{m}^{2}{n}^{2}
\\&+2\,{\pi }^{4}{d}^{2}
{n}^{4}+4\,{\pi }^{4}d{m}^{4}+8\,{\pi }^{4}d{m}^{2}{n}^{2}
+4\,{\pi }^{4}d{n}^{4}+2\,{\pi }^{4}{m}^{4}+4\,{\pi 
}^{4}{m}^{2}{n}^{2}\\&+2\,{\pi }^{4}{n}^{4}-2\,{L}^{4}\gamma^2\big)\alpha_i,
\end{split}
\end{equation*}
\begin{equation*}
\begin{split}
C_2(\alpha_i)=&15\,{L}^{4}\gamma^2{\alpha_i}^{4}-32\,{L}^{4}{\alpha_i}^{3}+\big(12\,{L}^{2}\gamma{\pi }^{2}d{m}^{2}+12\,{L}^{2}\gamma{\pi }^{2}d{n}^{2}-8\,{L}^{
6}+12\,{L}^{2}\gamma{\pi }^{2}{m}^{2}\\&+12\,{L}^{2}\gamma{
\pi }^{2}{n}^{2}\big){\alpha_i}^{2}+\big(-12\,{L}^{2}{\pi }^{4}d{m}^{4}-24\,{
L}^{2}{\pi }^{4}d{m}^{2}{n}^{2}-12\,{L}^{2}{\pi }^{4}d{n}^{4}
\\&-12\,{L}^{2}{\pi }^{4}{m}^{4}-24\,{L}^{2}{\pi }^{4}{m}^{2}
{n}^{2}-12\,{L}^{2}{\pi }^{4}{n}^{4}+4\,{L}^{4}{\pi }^{2}d
{m}^{2}+4\,{L}^{4}{\pi }^{2}d{n}^{2}\\ &+4\,{L}^{4}{\pi }^{2}{
m}^{2}+4\,{L}^{4}{\pi }^{2}{n}^{2}\big)\alpha_i+{\pi }^{4}{d}^{2}{m}^{4
}+2\,{\pi }^{4}{d}^{2}{m}^{2}{n}^{2}+{\pi }^{4}{d}^{2}{n}^{4}+2\,{\pi 
}^{4}d{m}^{4}\\&+4\,{\pi }^{4}d{m}^{2}{n}^{2}+2\,{\pi }^{4}d{n}^{4}+{\pi 
}^{4}{m}^{4}+2\,{\pi }^{4}{n}^{2}{m}^{2}+{\pi }^{4}{n}^{4}-2\,{L}^{2}\gamma{
\pi }^{2}d{m}^{2}\\ &-2\,{L}^{2}\gamma{\pi }^{2}d{n}^{2}-2\,{L}^{2}\gamma{\pi }^{2}{m}
^{2}-2\,{L}^{2}\gamma^2{\pi }^{2}{n}^{2}+{L}^{4}\gamma^2,
\end{split}
\end{equation*}
\begin{equation*}
\begin{split}
C_3(\alpha_i) =&20\,{L}^{4}\gamma^2{\alpha_i}^{3}-48\,{L}^{4}{\alpha_i}^{2}+\big(-8\,{L}^{4}{\pi }^{2}{m
}^{2}-8\,{L}^{4}{\pi }^{2}{n}^{2}+8\,{L}^{2}\gamma{\pi }^{2}d{m}
^{2}\\&+8\,{L}^{2}\gamma{\pi }^{2}d{n}^{2}+8\,{L}^{2}\gamma{\pi }^{2}{m}^
{2}+8\,{L}^{2}\gamma{\pi }^{2}{n}^{2}-4\,{L}^{4}\gamma\big)\alpha_i-4\,{L}^{2
}{\pi }^{4}d{m}^{4}\\&-8\,{L}^{2}{\pi }^{4}d{m}^{2}{n}^{2}-4\,{L}^{2}{
\pi }^{4}d{n}^{4}-4\,{L}^{2}{\pi }^{4}{m}^{4}-8\,{L}^{2}{\pi }^{4}{m}^
{2}{n}^{2}-4\,{L}^{2}{\pi }^{4}{n}^{4}\\&+4\,{L}^{4}{\pi }^{2}d{m}^{2}+4
\,{L}^{4}{\pi }^{2}d{n}^{2}+4\,{L}^{4}{\pi }^{2}{m}^{2}+4\,{L}^{4}{
\pi }^{2}{n}^{2},
\end{split}
\end{equation*}
\begin{equation*}
\begin{split}
C_4(\alpha_i)=&15\,{L}^{4}\gamma^2{\alpha_i}^{2}-32\,{L}^{4}\alpha_i-4\,{L}^{4}{\pi }^{2}{m}^{2}-
4\,{L}^{4}{\pi }^{2}{n}^{2}+2\,{L}^{2}\gamma{\pi }^{2}d{m}^{2}\\&+2\,{L}^{2}\gamma{
\pi }^{2}d{n}^{2}+4\,{L}^{6}+2\,{L}^{2}\gamma{\pi }^{2}{m}^{2}+2\,{L}^{2}\gamma{
\pi }^{2}{n}^{2}-2\,{L}^{4}\gamma^2,
\end{split}
\end{equation*}
\begin{equation*}
\begin{split}
C_5(\alpha_i)=6\,{L}^{4}\gamma^2\alpha_i-8\,{L}^{4}, \quad \text{and} \quad C_6(\alpha_i)=&\gamma^2L^4.
\end{split}
\end{equation*}
The implicit curve satisfying (\ref{star2}) indicates the choices of $(\alpha,\beta)$ for which the eigenvalues are repeated real roots, since this is the curve on which the discriminant is zero.  Therefore it forms the boundary region in the plane corresponding to complex eigenvalues. The polynomial $\psi_i(\beta)$  of degree 6, given by (\ref{star2}) is solved for $5000$ fixed $\alpha_i$ on a domain with $\alpha_{max}=\beta_{max}=3$. For each $\alpha_i$ the positive real roots of $\psi_i(\beta)$ are extracted and plotted on the $(\alpha,\beta)\in \mathbb{R}_+$ plane. The algorithm is run for five different values of the non-dimensional parameter $d$, associated to both cases, where the area  $L^2$ of $\Omega$ satisfies condition (\ref{cond}) of Theorem \ref{theorem1} as well as condition (\ref{cond1}) of Theorem \ref{theorem2}. 

The equation for the second curve that partitions the region corresponding to complex eigenvalues of the parameter plane is the curve on which the real parts of the complex roots are zero, hence satisfying the equation
\begin{equation}
\mathcal{T}(\alpha,\beta) =\gamma \frac{\beta-\alpha-(\alpha+\beta)^3}{\beta+\alpha}-\frac{(d+1)(m^2+n^2)\pi^2}{L^2}=0.
\label{realzero}
\end{equation}
The solution of (\ref{realzero}) by algebraic manipulation can be found to be equivalent to finding the positive real roots of the cubic polynomial $\phi(\alpha_i,\beta)=0$ for fixed $\alpha_i$, where $\phi$ is given by
 \begin{equation}
 \phi(\alpha_i,\beta) = C_0(\alpha_i)+C_1(\alpha_i)\beta+C_2(\alpha_i)\beta^2+C_3(\alpha_i)\beta^3,
 \label{impcen}
 \end{equation}
 with  $C_0(\alpha_i) = -\big(L^2\gamma+(d+1)(n^2+m^2)\pi^2\big)-L^2\gamma\alpha_i^3,$ 
$ C_1(\alpha_i) = L^2\gamma-(d+1)(n^2+m^2)\pi^2-3L^2\gamma\alpha_i^2,$
 $C_2(\alpha_i) = -3L^2\gamma\alpha_i,$ and 
 $C_3(\alpha_i) = -L^2\gamma$.
 %\end{split}\]
 
Variations of the parameter $d$ are investigated, whilst the area $L^2$ remains to satisfy condition (\ref{cond}) and as expected it is found that in the region of the parameter plane corresponding to complex eigenvalues, there is no sub-region that corresponds to $\lambda_{1,2}$ to have positive real parts. For each value of $d$ the region corresponding to complex eigenvalues is tested by looking for a critical curve on which $\lambda_{1,2}$ is purely imaginary, i.e. satisfying (\ref{impcen}). If such a curve exists, it would correspond to the system undergoing periodic oscillations around $(u_s,v_s)$, thus the system exhibiting transcritical bifurcation, which also implies the existence of a region in the parameter space that corresponds to real parts of $\lambda_{1,2}$ to be positive. Upon investigating this region, it is found that such a curve under condition (\ref{cond}) does not exists, and all roots corresponding to the cubic polynomial given by (\ref{impcen}) are either complex eigenvalues or they are negative real values, therefore cannot be the choice of admissible parameters of the system. Another observation is that as $d$ increases the area of the region corresponding to complex eigenvalues in the $(\alpha,\beta)$ plane gradually decreases. The domain size of length $L=5$ satisfying the condition of Theorem \ref{theorem1} is tested for stability analysis of $(u_s,v_s)$ when $\lambda_{1,2}$ are complex eigenvalues. The results are shown for different values of $d$ by a distinct colour in Figure \ref{fig6} (a). Finding the solutions of these curves does not however indicate, which side of them correspond to complex $\lambda_{1,2}$ and which side to real $\lambda_{1,2}$. This is decided by some numerical trial and error by evaluating $\lambda_{1,2}$ using few values from both sides of each curve. Trial and error indicates that the regions under these curves correspond to $\lambda_{1,2}$ to be a complex conjugate pair and hence any combination of $(\alpha,\beta)$ from this region ensures that the eigenvalues $\lambda_{1,2}$ contain a non-zero imaginary part. When the area $L^2$ of $\Omega$ satisfies condition (\ref{cond}), the absence of a region satisfying (\ref{impcen}) verifies the statement of Theorem \ref{theorem1}, therefore the region corresponding to complex eigenvalues for the choice of $L^2$ satisfying (\ref{cond}) has no sub-partitions, because everywhere in this region the real part of the eigenvalues is negative, hence no choice of parameters could result in the system to exhibit Hopf or transcritical bifurcation. The eigenvalues only become positive when they are both real values, therefore condition (\ref{cond}) restricts the diffusion-driven instability to Turing type only. Figure \ref{fig7} (a) shows how the partitioning curve changes location, as the value of $d$ is varied. The region under those curves corresponds to the complex eigenvalues. Similarly Figure \ref{fig7} (b) indicates the regions corresponding to complex eigenvalues for the corresponding values of $d$. Each stripe in Figure \ref{fig7} (b) is denoted by a letter that represents the set of all  points corresponding to a distinct colour stripe.  Similarly, the regions where the eigenvalues $\lambda_{1,2}$ are negative real roots are presented by Figure \ref{fig8} (a), corresponding to the same values of the parameter $d$. Figure \ref{fig8} (b) shows regions where at least one or both eigenvalues are positive real roots. The summary of Figures \ref{fig7} and \ref{fig8} is presented in Table \ref{Table2}. 
\begin{figure}[ht]
 \centering
 \small
  \begin{subfigure}[ht]{.47\textwidth}
    \centering
 \includegraphics[width=\textwidth]{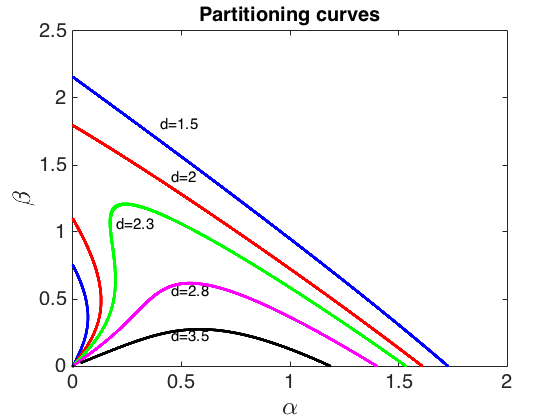}
 \caption{Boundary curves for complex $\lambda_{1,2}$\\ with corresponding values of $d$ and the \\condition $L^2=25<\pi^2\frac{(d+1)(m^2+n^2)}{\gamma}$.}
 \label{curves}
 \end{subfigure}
 \begin{subfigure}[ht]{.47\textwidth}
 \centering
    \includegraphics[width=\textwidth]{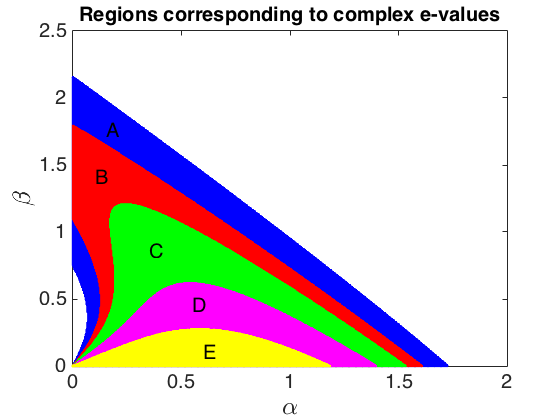}
  \caption{Regions corresponding to complex eigenvalues associated to values of $d$ indicated in Figure (a) and condition (\ref{cond}).}
  \label{regions}
 \end{subfigure}
 \caption{Parameter regions corresponding to complex eigenvalues and their boundary curves for various values of $d$ and domain size $L$ restricted to the condition (\ref{cond}) of Theorem \ref{theorem1}.}
  \label{fig7}
\end{figure}
 
\begin{figure}[ht]
 \centering
 \small
  \begin{subfigure}[ht]{.47\textwidth}
    \centering
 \includegraphics[width=\textwidth]{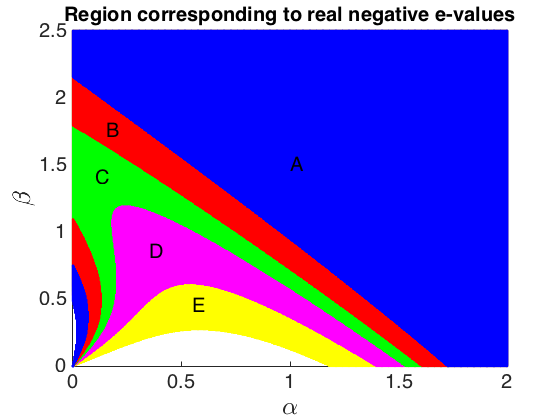}
 \caption{Regions where $\lambda_{1,2}$ are both real\\ and negative.}
 \label{negreal}
 \end{subfigure}
 \begin{subfigure}[ht]{.47\textwidth}
 \centering
    \includegraphics[width=\textwidth]{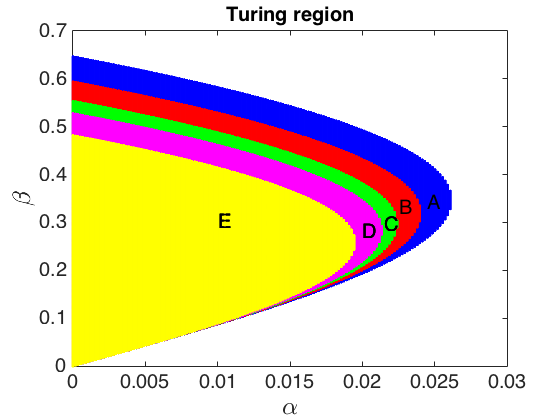}
  \caption{Regions where both eigenvalues are real and at least one of $\lambda_1$ or $\lambda_2$ are positive.}
  \label{posreal}
 \end{subfigure}
 \caption{The region where $\lambda_{1,2}$ are both real for various values of $d$ and domain size $L^2$ restricted to the condition (\ref{cond}) in Theorem \ref{theorem1}.}
  \label{fig8}
\end{figure}
\begin{table}[hbtp]
\centering
\small
\tabcolsep=0.3cm
\noindent\adjustbox{max width=\textwidth}{
\begin{tabular}{|c |c |c |c |c| c|}
\cline{1-1}
\hline
 Plot index& Figure \ref{fig8} (a) & Figure \ref{fig8} (b)& Figure \ref{fig7}  (b)&  Figure \ref{fig7} (b)& Figure \ref{fig7} (b)\\
\hline
 Eigenvalues& $0>\lambda_{1,2} \in \mathbb{R}$ & $0<\lambda_{1,2} \in \mathbb{R}$& $\lambda \in \mathbb{C}, Re(\lambda_{1,2})<0$& $\lambda \in \mathbb{C}, Re(\lambda_{1,2})>0$ &$\lambda\in\mathbb{C}, Re(\lambda)=0$\\
 \hline
\diaghead{\theadfont{\normalsize} Type of (SSSS) }{Value of $d$}{Type of (SS)}&Stable node& Turing type instability& Stable spiral&Hopf bifurcation&Transcritical bifurcation\\
\hline
1.5 & $A$ & $E$ & $A \cup B\cup C \cup D \cup E$ & $\emptyset$ & $\emptyset$\\
\hline
2 & $A \cup B$ & $E \cup D$& $B\cup C \cup D \cup E$ & $\emptyset$ & $\emptyset$\\
\hline
2.3  & $A \cup B\cup C$ & $E \cup D \cup C$ &  $C \cup D \cup E$ & $\emptyset$& $\emptyset$\\
\hline
2.8  & $A \cup B\cup C \cup D$ & $E \cup D \cup C \cup B$   &   $D \cup E$ & $\emptyset$& $\emptyset$\\
\hline
3.5 & $A \cup B\cup C \cup D \cup E$ & $E \cup D \cup C \cup B \cup A$ &  $E$ & $\emptyset$& $\emptyset$\\
\hline
\end{tabular}}
\caption{Table showing regions corresponding to domain size $L^2 < \frac{\pi^2(d+1)(m^2+n^2)}{\gamma}$, which satisfies (\ref{cond}) of Theorem \ref{theorem1}.}
\label{Table2}
\end{table}
The algorithm is also run for the case when $L^2$ is chosen such that it satisfies (\ref{cond1}) and we find that the region corresponding to complex eigenvalues is further partitioned by the curve satisfying (\ref{realzero}). This is the curve on which the eigenvalues are purely imaginary. This curve also indicates that within the region corresponding to complex eigenvalues there is a sub-region in which the eigenvalues are complex with positive real part, which corresponds to the system exhibiting Hopf bifurcation. For choices of parameter values on the curve the system is expected to undergo transcritical bifurcation. The area of the domain $\Omega$ is taken as $L^2 = 225$ in order to satisfy the condition given by Theorem \ref{theorem2} with respect to $d$ and $\gamma$, so that there is enough space for varying $d$ and yet maintaining condition (\ref{cond1}).
It can be easily observed from the parameter space classification that there is a relatively small region in each case corresponding to diffusion-driven instability. This indicates the importance of making sure to choose parameter choices wisely, in order to expect the dynamics of the system to evolve to certain types of patterns or maybe no patterns at all. Table \ref{Table3} presents the summary for how the regions of the parameter space change with varying the parameter $d$. It would be reasonable to use the same variation of the parameter $d$ in both cases of the domain sizes, however, in the first case where the domain size satisfies (\ref{cond}) the span of varying parameter $d$ is relatively small, yet causing significant observable change in the parameter space. However, when the domain size is chosen according to condition (\ref{cond1}), small variations in the diffusion coefficient makes insignificant changes to the parameter spaces, therefore, in order to pictorially observe the consequential change in the parameter plane $(\alpha,\beta)$ the span of variations for $d$ have to be significantly large as seen in Table \ref{Table2}.
\begin{figure}[ht]
 \centering
 \small
  \begin{subfigure}[ht]{.47\textwidth}
    \centering
 \includegraphics[width=\textwidth]{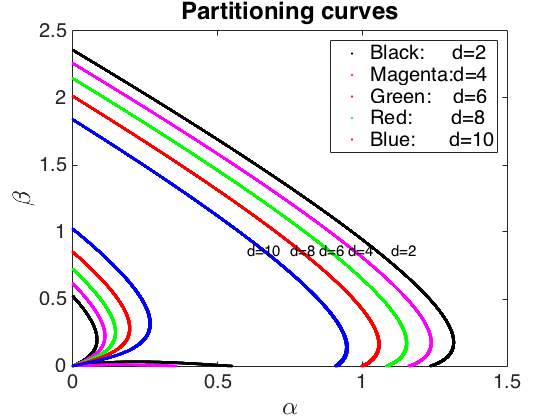}
 \caption{Boundary curves for complex $\lambda_{1,2}$\\ with corresponding values of $d$ and \\the condition $L^2=225>\pi^2\frac{(d+1)(m^2+n^2)}{\gamma}$.}
 \label{curvesthm2}
 \end{subfigure}
 \begin{subfigure}[ht]{.47\textwidth}
 \centering
    \includegraphics[width=\textwidth]{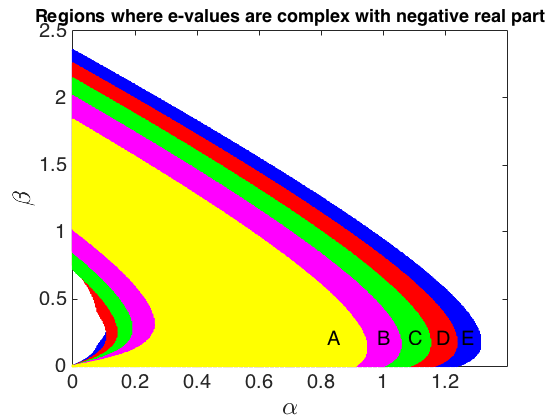}
  \caption{Regions corresponding to complex eigenvalues associated to values of $d$ indicated in Figure (a) and condition (\ref{cond1}).}
  \label{regionsthm2}
 \end{subfigure}
 \caption{Parameter regions corresponding to complex eigenvalues and their boundary curves for various values of $d$ and domain size $L$ restricted to the condition (\ref{cond1}) of Theorem \ref{theorem2}.}
  \label{fig9}
\end{figure}
\begin{figure}[H]
 \centering
 \small
  \begin{subfigure}[ht]{.47\textwidth}
    \centering
 \includegraphics[width=\textwidth]{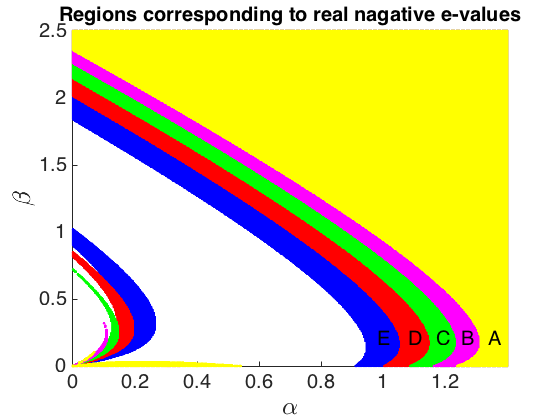}
 \caption{Regions where $\lambda_{1,2}$ is real distinct \\and negative.}
 \label{stabregion}
 \end{subfigure}
 \begin{subfigure}[ht]{.47\textwidth}
 \centering
    \includegraphics[width=\textwidth]{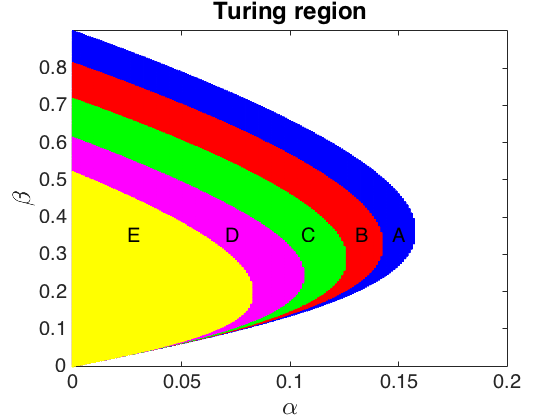}
  \caption{Region where $\lambda_{1,2}$ are real distinct and at least one of them positive.} 
  \label{unstabregion}
 \end{subfigure}
 \caption{Different colours denote different regions of stability of the model system  \eqref{B}-\eqref{zerofluxBCs} for different values of $d$.}
  \label{fig10}
\end{figure}
\begin{figure}[ht]
 \centering
 \small
  \begin{subfigure}[ht]{.49\textwidth}
    \centering
 \includegraphics[width=\textwidth]{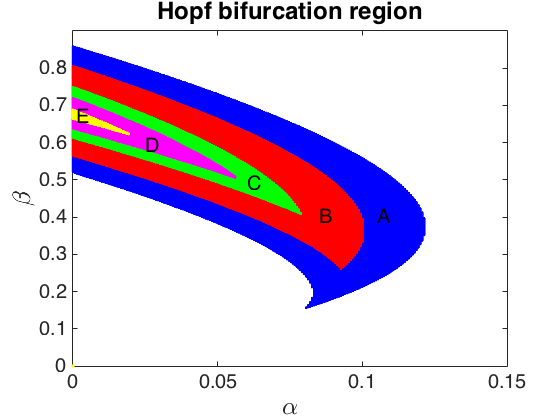}
 \caption{Region for complex $\lambda_{1,2}$ and\\ positive real part with corresponding\\ values of $d$ and the condition\\ $L^2=225>\pi^2\frac{(d+1)(m^2+n^2)}{\gamma}$.}
 \label{diffdrivthm2}
 \end{subfigure}
 \begin{subfigure}[ht]{.49\textwidth}
 \centering
    \includegraphics[width=\textwidth]{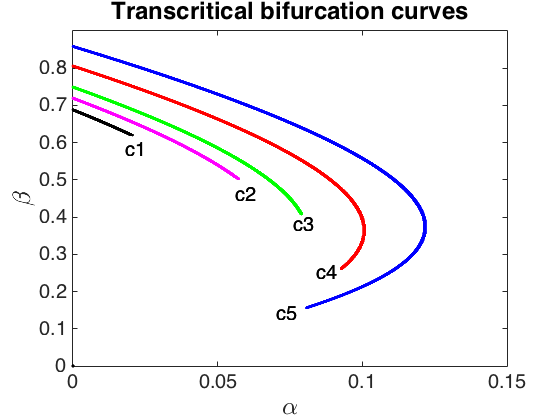}
  \caption{Regions corresponding to complex eigenvalues associated to values of $d$ indicated in Figure (a) and condition (\ref{cond1}).}
  \label{transc}
 \end{subfigure}
 \caption{Parameter regions corresponding to complex eigenvalues and their boundary curves for various values of $d$ and domain size $L$ restricted to the condition (\ref{cond1}) of Theorem \ref{theorem2}.}
  \label{fig11}
\end{figure}
\begin{table}[hbtp]
\centering
\small
\tabcolsep=0.3cm
\noindent\adjustbox{max width=\textwidth}{
\begin{tabular}{|c |c |c |c |c| c|}
\cline{1-1}
\hline
 Plot index& Figure \ref{fig10} (a) & Figure \ref{fig10} (b)& Figure \ref{fig9}  (b)&  Figure \ref{fig11} (a)& Figure \ref{fig11} (b)\\
\hline
 Eigenvalues& $0>\lambda_{1,2} \in \mathbb{R}$ & $0<\lambda_{1,2} \in \mathbb{R}$& $\lambda \in \mathbb{C}, Re(\lambda_{1,2})<0$& $\lambda \in \mathbb{C}, Re(\lambda_{1,2})>0$ &$\lambda\in\mathbb{C}, Re(\lambda)=0$\\
 \hline
\diaghead{\theadfont{\normalsize} Type of (SSSS) }{Value of $d$}{Type of (SS)}&Stable node& Turing type instability& Stable spiral&Hopf bifurcation&Transcritical bifurcation\\
\hline
2 & $A$ & $E$ & $A \cup B\cup C \cup D \cup E$ & $A \cup B\cup C \cup D \cup E$ & $c_5$\\
\hline
4 & $A \cup B$ & $E \cup D$& $B\cup C \cup D \cup E$ & $B\cup C \cup D \cup E$ & $c_4$\\
\hline
6  & $A \cup B\cup C$ & $E \cup D \cup C$ &  $C \cup D \cup E$ &$C \cup D \cup E$& $c_3$\\
\hline
8  & $A \cup B\cup C \cup D$ & $E \cup D \cup C \cup B$   &   $D \cup E$ & $D \cup E$& $c_2$\\
\hline
10 & $A \cup B\cup C \cup D \cup E$ & $E \cup D \cup C \cup B \cup A$ &  $E$ & $E$& $c_1$\\
\hline
\end{tabular}}
\caption{Table showing regions corresponding to the domain size $L^2 = 225 \geq \frac{\pi^2(d+1)(m^2+n^2)}{\gamma}$, which satisfies (\ref{cond1}) of Theorem \ref{theorem2}.}
\label{Table3}
\end{table}

In the case when $L^2$ satisfies Theorem \ref{theorem1}, it can be observed that the type of diffusion-driven instability that can occur is restricted to Turing type only, which is increased as the non-dimensional diffusion coefficient $d$ was increased. Turing diffusion-driven instability also occurs for the case when $L^2$ satisfies condition (\ref{cond1}) and it also increases with increased value of $d$ as shown in Figure \ref{fig10} (b). The interpretation is that regions in the parameter space exist, which result in the system to be stable in the absence of diffusion, but when diffusion is added to the system, the choice of the parameters from these particular regions result in the system exhibiting instability. This type of instability is restricted to space and hence leads to spatial patterning only, because the eigenvalues are both real. However, if $L^2$ satisfies (\ref{cond1}), then in addition to the existence of regions of the parameter space corresponding to Hopf and transcritical bifurcations, it can also be observed that with increased values of parameter $d$, unlike Turing instability, the regions for Hopf and transcritical bifurcations reduce as shown in Figure \ref{fig11} (a) and (b). The mathematical intuition for this inverse relation is the direct proportionality of $d$ on the lower bound for the domain size i.e. the right-hand side of (\ref{cond1}). Therefore, as $d$ grows, one gets closer to the violation of the necessary condition (\ref{cond1}) for Hopf and transcritical bifurcations as proposed by Theorem \ref{theorem2}. However, the region for Hopf and transcritical bifurcations remains to exist as long as $L^2$ satisfies (\ref{cond1}), which ultimately means that for the system to exhibit temporal periodicity (patterning in time) the domain size has to be sufficiently large satisfying (\ref{cond1}).
\subsubsection{Remark}
 The eigenfunction modes used for the current simulations corresponds to the first non-trivial (non-zero) eigenvalues, namely the case where $m=n=1$. The results can be readily obtained for any positive integer values of $m$ and $n$. Despite the fact that using larger integer values for $m$ and $n$ will shift or scale the classification for the parameter space, the conditions given by Theorem \ref{theorem1} and Theorem \ref{theorem2} will remain intact irrespective of the values of $m$ and $n$, so long as they are positive integers.

\subsection{Numerical experiments}
To validate the proposed classification of parameter space, the reaction-diffusion system (\ref{B})-\eqref{ICs} is numerically solved using the finite element method \cite{paper3, paper8, paper12, book2} on a unit square domain with a uniform triangular mesh.
Numerical simulations are performed for various choices of parameter values $\alpha$ and $\beta$, chosen from the appropriate parameter spaces to demonstrate and validate our theoretical findings. In all our simulations, we vary the parameters $d$ and $\gamma$ and keep fixed the domain length size $L$ and this allows us to keep constant the well refined number of degrees of freedom for the mesh. 
%The domain-size of a unit square is used for all the simulations, where reaction and diffusion rates are varied, for the conditions stated in Theorem \ref{theorem1} and 2 to be obtained. One reason why all simulations are performed on a unit square is because if domain size is varied for fixed parameters $d$ and $\gamma$, it will add a significant amount of computational load associated with maintaining equivalent refinement of the mesh. Therefore, keeping the domain size fixed and varying parameters $d$ and $\gamma$ comes with no extra computational cost yet producing the desired verification for our theoretical predictions. 
The initial conditions for each simulation are taken as small random perturbations around the neighbourhood of the steady state of the form \cite{paper15, paper29}
\begin{equation}
\begin{cases}
u_0(x,y)=\alpha+\beta+0.0016\cos(2\pi(x+y))+0.01\sum_{i=1}^8\cos(i\pi x),\\
v_0(x,y)= \frac{\beta}{(\alpha+\beta)^2}+0.0016\cos(2\pi(x+y))+0.01\sum_{i=1}^8\cos(i\pi x).
\end{cases}
\label{initial}
\end{equation}
The final time $T$ is chosen to ensure that beyond $T$ the convergence of solutions in successive time-step differences decay to a threshold of $10^{-5}$ for all the simulations. Therefore, in the case where the area $L^2$ of the domain size satisfies condition (\ref{cond}) proven in Theorem \ref{theorem1}, when temporal instability is forbidden by the condition on the domain size, the final time of numerical simulations is relatively shorter compared to that used for other cases. For numerical simulations on the domain size $L^2$ satisfying condition (\ref{cond1}) proven in Theorem \ref{theorem2}, the final time is experimented for longer periods to capture possible existence of temporal periodicity in the dynamics. Simulations with domain size satisfying (\ref{cond1}) are captured at different times, which varies with cases, for which details are included in Table \ref{Table4}.
In all our numerical results, we only exhibit numerical solutions corresponding to the $u(x,y,t)$ component, those of $v(x,y,t)$ are known to be 180-degrees out of phase to those of $u$. 

Figure \ref{fig14} presents the case where the parameters $\alpha=2$ and $\beta=2$ are chosen from stable node region of Table \ref{Table2}, with final time $T=2$.  It can be observed that the evolved profile of the concentration $u$ uniformly converges to the steady state value forming neither spatial nor temporal patterns (Figure \ref{fig14}  (a)). As predicted in Theorem \ref{theorem1}, no choice of parameters $\alpha$ and $\beta$ can influence the dynamics to exhibit temporal periodicity, therefore any choice of the parameters outside Turing space, given that the domain size satisfies (\ref{cond}) will uniformly converge to the stable steady state $(u_s,v_s)=(\alpha+\beta, \frac{\beta}{(\alpha+\beta)^2})  \approx (4,0.125)$ as seen in Figure \ref{fig14} (a). Figure \ref{fig14} (b) shows the uniform convergence of the discrete $L_2$-norm difference between solutions at successive time-steps to the constant steady state $(u_s,v_s)$, where there is no sign of instability occurring, when parameters are outside the Turing space.    
\begin{figure}[H]
 \centering
 \small
  \begin{subfigure}[h]{.49\textwidth}
    \centering
 \includegraphics[width=\textwidth]{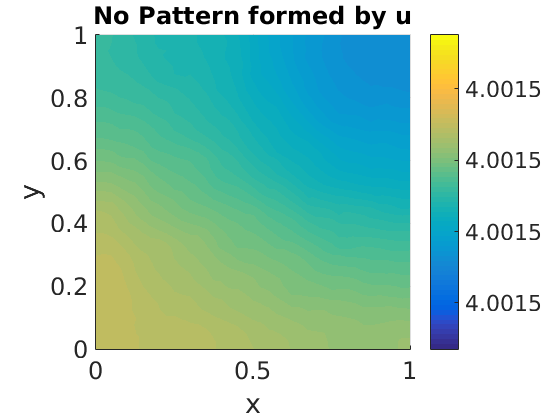}
 \caption{The evolved discrete solution $U$ at\\ the final time step $T=2$.}
 \label{pattern2}
 \end{subfigure}
 \begin{subfigure}[h]{.49\textwidth}
 \centering
    \includegraphics[width=\textwidth]{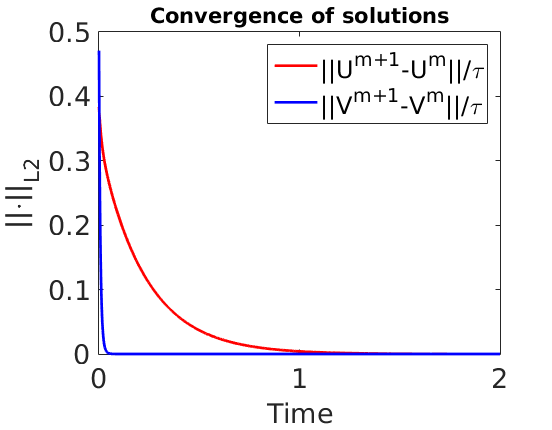}
  \caption{Convergence of the difference of the discrete solution $U$ and $V$ for successive time steps. }
  \label{Error2}
 \end{subfigure}
 \caption{Finite element numerical simulations corresponding to the $u$ of the model system \eqref{B}-\eqref{zerofluxBCs} with parameter values $\alpha$ and $\beta$ selected from outside  the Turing space (see Table \ref{Table4} for values) and domain size satisfying (\ref{cond}). No patterns are obtained in agreement with the theoretical predictions. }
  \label{fig14}
\end{figure}
%Figure \ref{fig14} shows the finite element solution for the variable $u$, which at the final time $T=2$ already attains the constant uniform steady state $u_s \approx 4$, which is a consequence of selecting parameters outside the unstable region in the parameter space i.e. $(\alpha,\beta)=(2,2)$. The finite element solution $v$ also converges to a constant value of the uniform steady state, which is $v_s\approx \frac{1}{8}$.  
\begin{figure}[H]
 \centering
 \small
  \begin{subfigure}[h]{.49\textwidth}
    \centering
 \includegraphics[width=\textwidth]{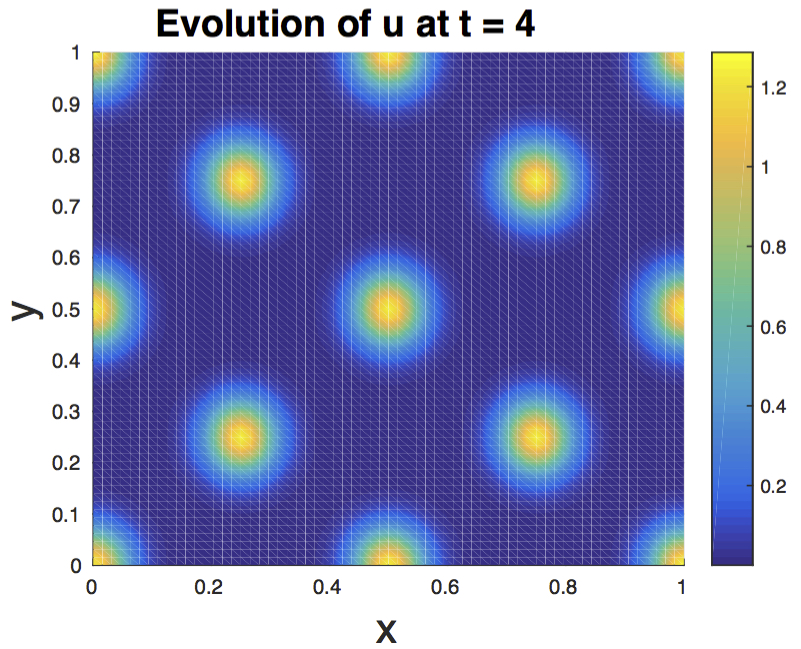}
 \caption{The evolved discrete solution \\$U$ at the final time step $T=4$.}
 \label{pattern1}
 \end{subfigure}
 \begin{subfigure}[h]{.49\textwidth}
 \centering
    \includegraphics[width=\textwidth]{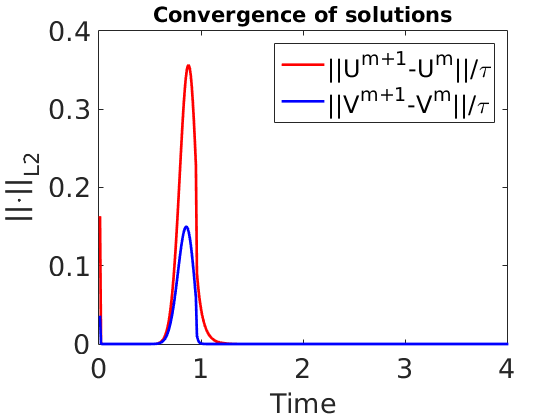}
  \caption{Convergence of the difference of the discrete solution $U$ and $V$ for successive time steps. }
  \label{Error1}
 \end{subfigure}
 \caption{Finite element numerical simulations corresponding to the $u$ of the model system \eqref{B}-\eqref{zerofluxBCs} with parameter values $\alpha$ and $\beta$ selected from  the Turing space (see Table \ref{Table4} for values) and domain size satisfying (\ref{cond}).  We observe the formation of spot patterns, again in agreement with the theoretical predictions.}
  \label{fig13}
\end{figure}
Figure \ref{fig13} (a) presents the evolved profile of the solution captured at the final time step, for the choice of parameters from Turing space presented in Table \ref{Table4}. The convergence in the discrete $L_2$-norm of difference in the solutions for the successive time steps from start to the end of the simulations are also plotted against time and presented in Figure \ref{fig13} (b).
\begin{figure}[H]
 \centering
 \small
  \begin{subfigure}[h]{.49\textwidth}
    \centering
 \includegraphics[width=\textwidth]{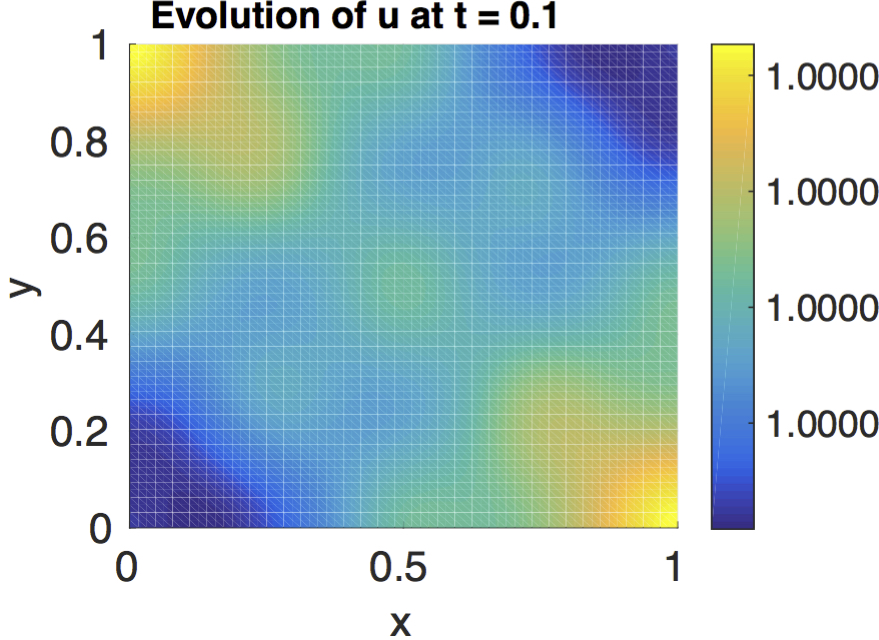}
 \caption{The evolving discrete solution\\ $U$   captured at time step $t=0.1$.}
 \label{pattern2aone}
 \end{subfigure}
 \begin{subfigure}[h]{.49\textwidth}
 \centering
    \includegraphics[width=\textwidth]{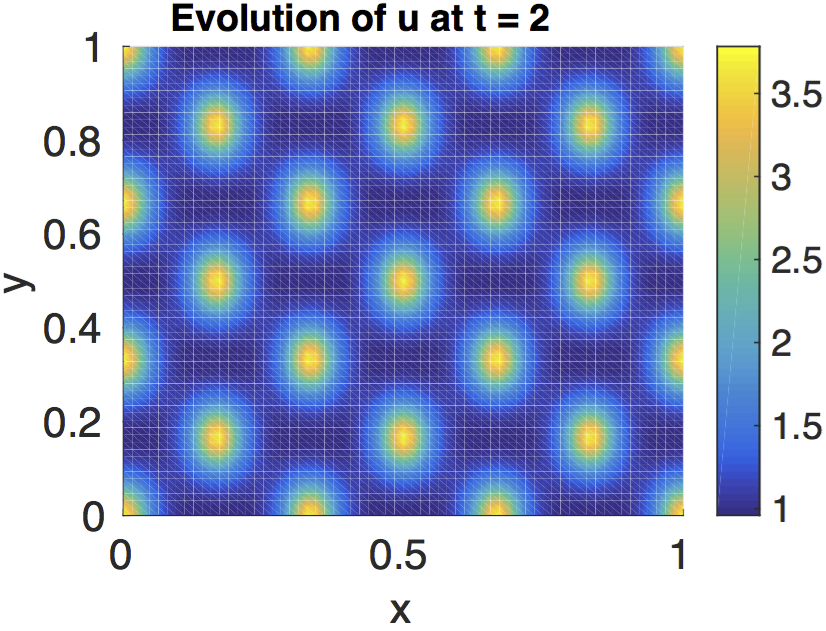}
  \caption{The evolving discrete solution $U$ \\captured at time step $t=2$. }
  \label{pattern2atwo}
 \end{subfigure}
 \begin{subfigure}[h]{.49\textwidth}
    \centering
 \includegraphics[width=\textwidth]{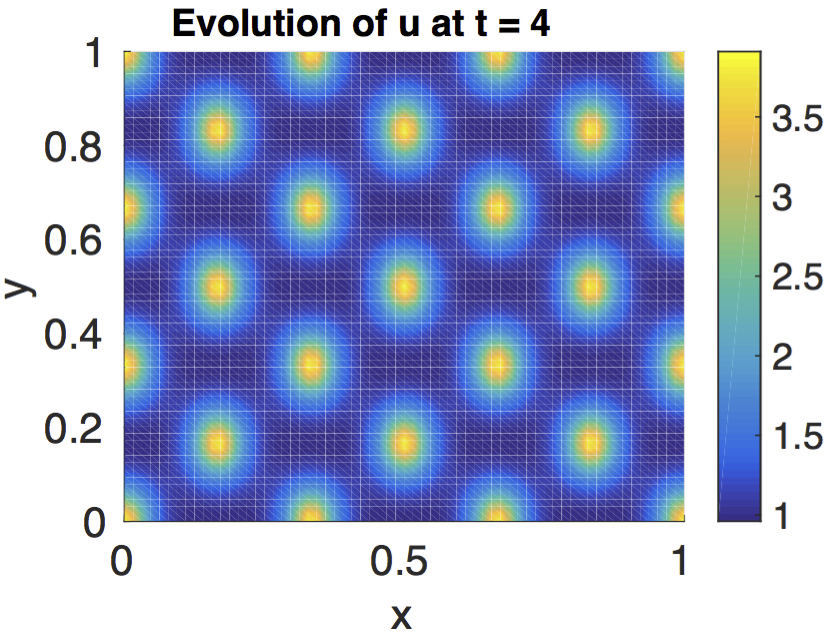}
 \caption{The evolved discrete solution $U$ \\captured at final time step $t=4$. }
 \label{pattern2athree}
 \end{subfigure}
 \begin{subfigure}[h]{.49\textwidth}
 \centering
    \includegraphics[width=\textwidth]{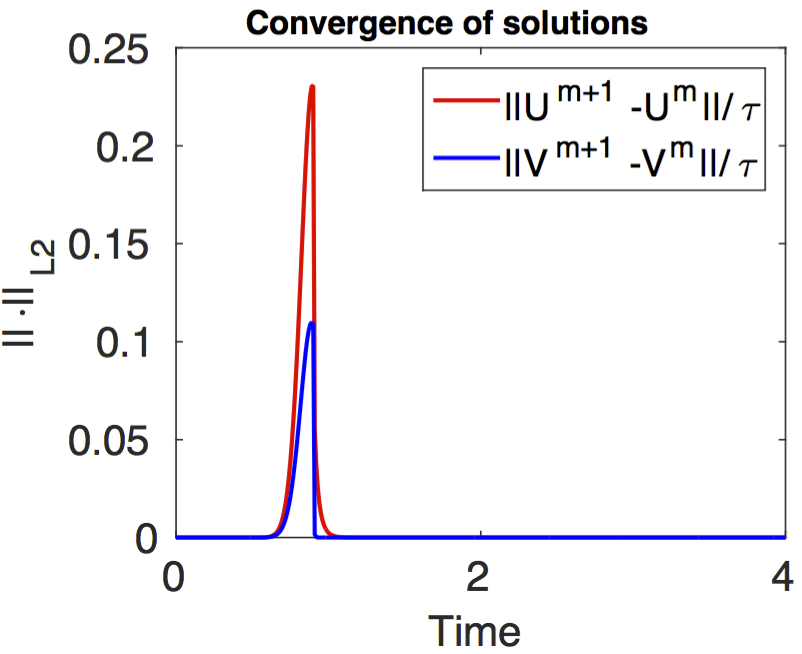}
  \caption{Convergence of the difference of the \\discrete solution $U$ and $V$ for successive\\ time steps. }
  \label{convergence3}
 \end{subfigure}
 \caption{Finite element numerical simulations corresponding to the $u$ of the model system \eqref{B}-\eqref{zerofluxBCs} with parameter values $\alpha$ and $\beta$ selected from  the Turing space (see Table \ref{Table4} for values) and domain size satisfying (\ref{cond1}). We observe the formation of spot patterns which are more clustered than those obtained in Figure \ref{fig13}.}
  \label{fig15}
\end{figure}
Figure \ref{fig15} shows the simulation for the choice of parameters from the Turing space presented in Table \ref{Table3} for the domain size satisfying (\ref{cond1}). In this case as predicted in Theorem \ref{theorem2}, regions of parameter space exist for which the dynamic can exhibit Hopf type bifurcation, therefore, the possibility of temporal periodicity in the dynamics. Figure \ref{fig16} presents such periodicity in time for spatial patterns. The relative discrete $L_2$-norm of the difference in the solutions for successive time-steps is therefore showing time periodicity as illustrated in Figure \ref{fig16} (d), which indicates the transition of the solution from the initially achieved spatial pattern to a different spatial pattern.
\begin{figure}[H]
 \centering
 \small
  \begin{subfigure}[h]{.49\textwidth}
    \centering
 \includegraphics[width=\textwidth]{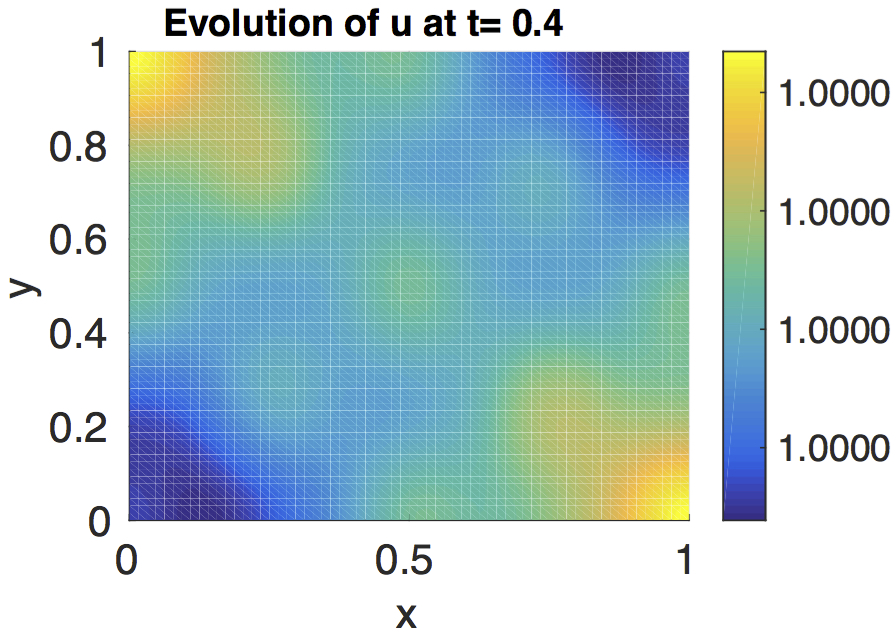}
 \caption{The evolving discrete solution\\ $U$   captured at time step $t=0.4$.}
 \label{pattern2bone}
 \end{subfigure}
 \begin{subfigure}[h]{.49\textwidth}
 \centering
    \includegraphics[width=\textwidth]{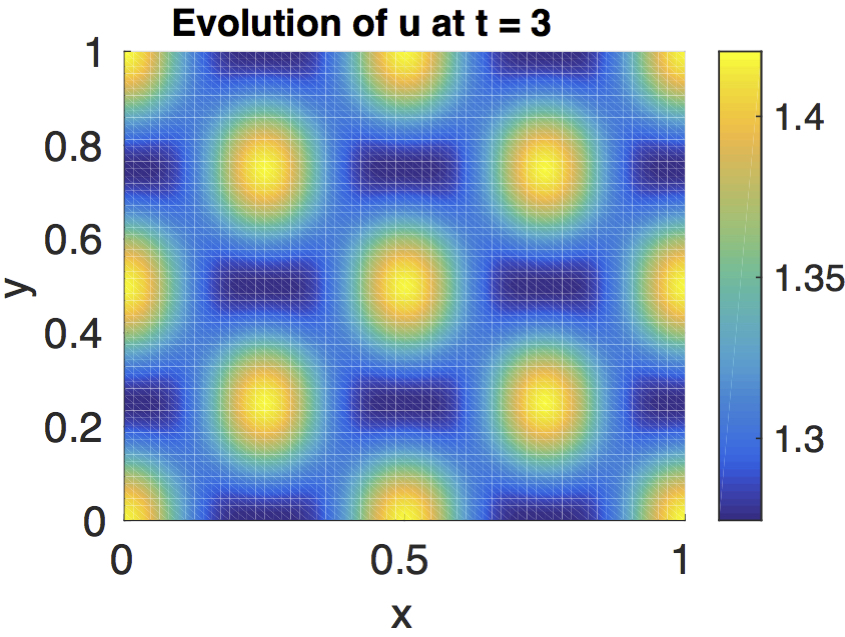}
  \caption{The evolving discrete solution $U$ \\captured at time step $t=3$. }
  \label{pattern2btwo}
 \end{subfigure}
 \begin{subfigure}[h]{.49\textwidth}
    \centering
 \includegraphics[width=\textwidth]{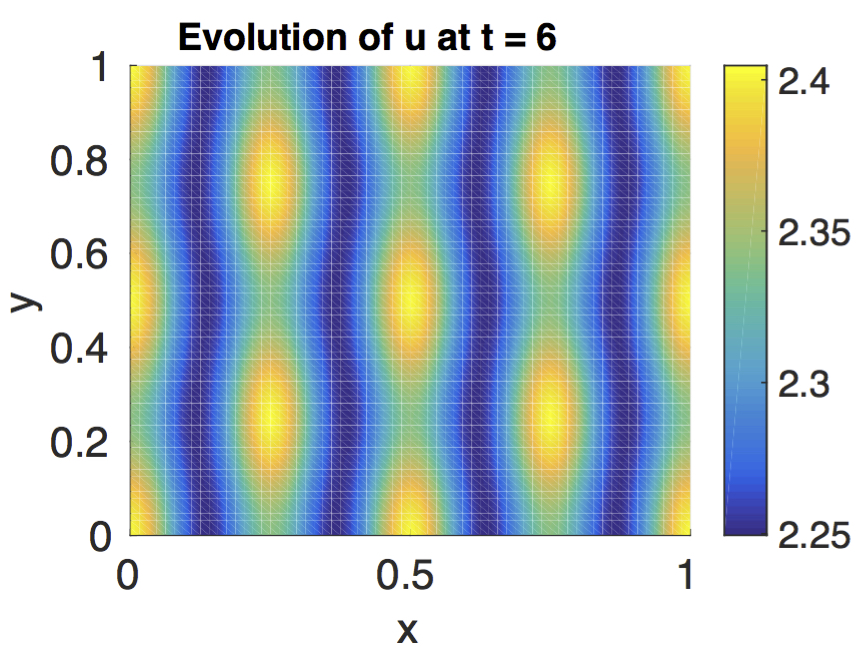}
 \caption{The evolved discrete solution $U$ \\captured at final time step $t=6$. }
 \label{pattern2bthree}
 \end{subfigure}
 \begin{subfigure}[h]{.49\textwidth}
 \centering
    \includegraphics[width=\textwidth]{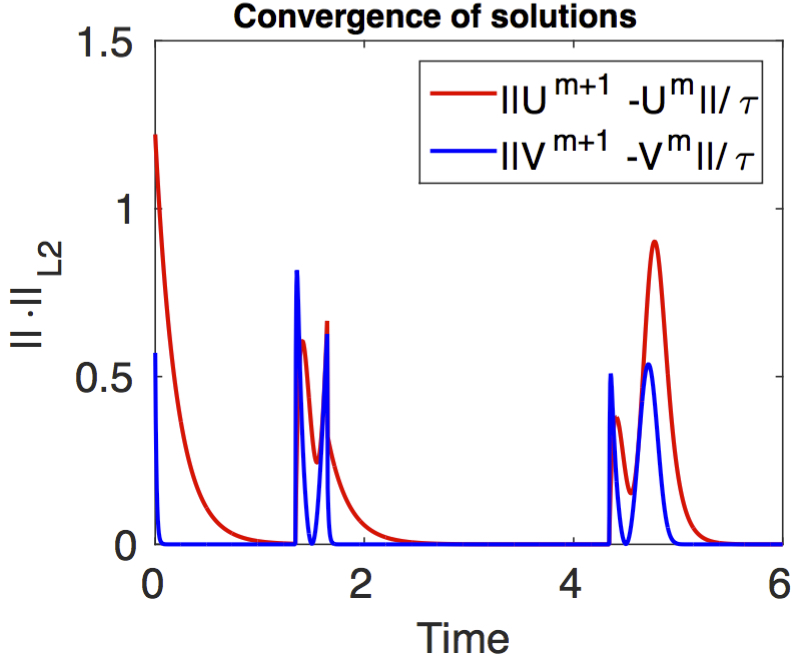}
  \caption{Convergence of the difference of the \\discrete solution $U$ and $V$ for successive\\ time steps. }
  \label{convergence4}
 \end{subfigure}
 \caption{Finite element numerical simulations corresponding to the $u$ of the model system \eqref{B}-\eqref{zerofluxBCs} with parameter values $\alpha$ and $\beta$ selected from the Hopf/Transcritical bifurcation region (see Table \ref{Table4} for values) and domain size satisfying (\ref{cond1}). We observe the formation of spatial-temporal periodic patterning in agreement with theoretical predictions}
  \label{fig16}
\end{figure}
%\begin{table}[hbtp]
\begin{table}[H]
\centering
\small
\tabcolsep=0.3cm
\noindent\adjustbox{max width=\textwidth}{
\begin{tabular}{|c |c |c |c |c|}
\cline{1-1}
\hline
 Plot index& Figure \ref{fig14}  & Figure \ref{fig13} & Figure \ref{fig15} &  Figure \ref{fig16}\\
\hline
\diaghead{\theadfont{\normalsize} Type of (SSSS) }{Parameters}{Instability}&\shortstack{No instability\\ No pattern}& \shortstack{Turing type instability\\Spatial pattern}& \shortstack{Turing type instability\\Spatial pattern}&\shortstack{Hopf bifurcation\\Spatial and temporal pattern}\\
\hline
$(\alpha,\beta)$ & $(2,2)$ & $(0.025, 0.4)$ & $(0.1, 0.6)$ & $(0.05, 0.8)$ \\
\hline
$(d, \gamma)$ & $ (10, 210)$ &$(10, 210)$& $ (10, 300)$ & $ (10,350 )$ \\
\hline
Condition on $\Omega$ & (\ref{cond}) & (\ref{cond}) &  (\ref{cond1}) &(\ref{cond1})\\
\hline
Simulation time & $2$ & $4$ &  $4$ & $6$\\
\hline
CPU time (sec) & $77.14$ & $149.60$ &  $149.55$ & $215.02$\\
\hline
\end{tabular}}
\caption{Showing the choice of parameters $(\alpha,\beta)$ for each simulation and the choice of $(d,\gamma)$ subject to the relevant condition referred to in third row. Each simulation was run with time-step of $1\times10^{-3}$.} 
\label{Table4}
\end{table}
\section{Conclusion}\label{conclusion}
In this work the full parameter space for a reaction-diffusion system with {\it activator-depleted} reaction kinetics was classified by use of linear stability theory and in each region of the parameter space the dynamics of the reaction-diffusion system was investigated. In the absence of diffusion, theoretical results on the dynamics of the system were supported by use of the phase-plane analysis, where in each case the numerical solution of the system was observed to be in agreement with the theoretically predicted behaviour. 
In the presence of diffusion, for a two-component reaction-diffusion system, two conditions relating the domain size to the diffusion and reaction rates were derived. The proofs of these conditions were presented in Theorems \ref{theorem1} and \ref{theorem2}, respectively. For full classification of the parameter space, a numerical method was used to compute the solutions of the implicit curves in the parameter space forming the partitions of classification. 
%The numerical method was applied under both conditions, where the results of classification under each condition was in agreement with the theoretical predictions given in Theorems \ref{theorem1} and \ref{theorem2}. 
In particular, using condition (\ref{cond}) the numerical method for solving the partitioning curves showed  the non-existence of a region in parameter space that (if existed) would lead to Hopf bifurcation or transcritical instability. Similarly, applying the numerical method to compute the partitioning curves under condition (\ref{cond1}), it was shown that regions in parameter space exist for both Hopf and transcritical bifurcation. Parameters from Hopf bifurcation region as well as from Turing regions under both conditions on the domain size were shown to be in agreement with the theoretical prediction, when the system was numerically solved by using the finite element method. For each simulation the discrete $L_2$-norm of the successive time-step difference of the solutions is also given to visualise the temporal dynamics of the behaviour of the solutions during the  convergence process.
\subsection{Ideas for future work}
This work sets the premises to study more complex systems of non-autonomous reaction-diffusion equations during growth development whereby the parameter spaces are continuously evolving due to domain growth. It will be revealing to study how the Turing diffusion-driven parameter space, the Hopf and Transcritical regions evolve with time and how the dynamics of the model system evolve. The motivation for such extension is to find whether the conditions (\ref{cond}) and (\ref{cond1}) for the domain size with reaction and diffusion rates continue to hold or whether a threshold for the domain size exists, beyond which, these conditions can be invalidated. Another direction of further investigation of the current work is to try and find similar relationships for different geometries such as spherical, elliptical and cylindrical domains. It may be noted that the eigenvalues and eigenfunctions of the Laplace operator change significantly depending on the boundary conditions and the geometry of the domain, therefore, application of the present method suggests, that finding similar conditions for other geometries require a careful step-by-step procedure to find analogous conditions to those given by (\ref{cond}) and (\ref{cond1}). It is also possible to apply the idea of the present work to the problems of pattern formation on bulk-surface geometries, where the relationship of surface area with reaction and diffusion rates can be explored to find the relevant influence of this relationship on the formation of spatial and/or temporal patterns. In \cite{paper20} it is shown that the known condition of $d>1$ is no longer necessary for pattern formation in the presence of linear cross-diffusion, which also provides a possible platform for the extension of the current work, to see, whether the invalidation of the condition $d>1$ in the presence of linear cross-diffusion has any dependence with the relationship of reaction and diffusion rates with the domain size.        
\section{Acknowledgments}
WS acknowledges support of the EPSRC Doctoral Training studentship. AM acknowledges support from the Leverhulme Trust Research Project Grant (RPG-2014-149), the European Union's Horizon 2020 research and innovation programme under the Marie Sklodowska-Curie grant agreement No 642866, and his work was partially supported by the Engineering and Physical Sciences Research Council, UK grant (EP/J016780/1). The authors (WS, AM) thank the Isaac Newton Institute for Mathematical Sciences for its hospitality during the programme (Coupling Geometric PDEs with Physics for Cell Morphology, Motility and Pattern Formation; EPSRC EP/K032208/1).  AM was partially supported by Fellowships from the Simons Foundation. AM is a Royal Society Wolfson Research Award Holder, generously supported by the Wolfson Trust.

%% The Appendices part is started with the command \appendix;
%% appendix sections are then done as normal sections
%\appendix
%
%\section{Section in Appendix}
%\label{appendix-sec1}
%
%Sample text. Sample text. Sample text. Sample text. Sample text. Sample text. 
%Sample text. Sample text. Sample text. Sample text. Sample text. Sample text. 
%Sample text. 
%\cite{paper1, paper2, paper3, paper4, paper5, paper6, paper7, paper8, paper9, paper10, paper11, paper12, paper13, paper14, paper15, paper16, paper17, paper18, paper19, paper20, paper21, paper22, paper23, paper24, paper25, paper26, paper27, paper28, paper29, paper30, paper31, paper32, book1, book2, book3, book4,book5, book6, book7, thesis1, thesis2}

%% References
%%
%% Following citation commands can be used in the body text:
%% Usage of \cite is as follows:
%%   \cite{key}         ==>>  [#]
%%   \cite[chap. 2]{key} ==>> [#, chap. 2]
%%

%% References with bibTeX database:

\bibliographystyle{elsarticle-num}

\bibliography{sample}

\end{document}